\documentclass[conference]{IEEEtran}
\usepackage{amsmath,graphicx, bm,amsfonts,dblfloatfix,amsthm}
\usepackage{cases}
\usepackage{subcaption}
\usepackage{accents}
\usepackage{color}
\usepackage{cite}
\usepackage{array}
\usepackage[acronym,nomain]{glossaries}
\usepackage[ruled,vlined,linesnumbered]{algorithm2e}
\usepackage[bottom]{footmisc}
\usepackage{pgfplots, tikz}
\pgfplotsset{
    compat=newest
}
\usetikzlibrary{positioning}
\DeclareMathOperator*{\Tr}{tr}

\newtheorem{lemma}{Lemma}

\newacronym{spucpa}{SPUCPA}{stacked polarimetric uniform circular patch array}
\newacronym{suca}{SUCA}{stacked uniform circular array}
\newacronym{scf}{SCF}{spatial correlation function}
\newacronym{uca}{UCA}{uniform circular array}
\newacronym{ula}{ULA}{uniform linear array}
\newacronym{doa}{DoA}{direction of arrival}
\newacronym{mse}{MSE}{mean squared error}
\newacronym{rmse}{RMSE}{root mean squared error}
\newacronym{cs}{CS}{compressed sensing}
\newacronym{omp}{OMP}{orthogonal matching pursuit}
\newacronym{sgd}{SGD}{Stochastic Gradient Descent}
\newacronym{crb}{CRLB}{Cramér-Rao Lower Bound}

\newcommand{\comment}[1]{}

\newcommand{\trans}{{\rm T}}

\newcommand{\herm}{{\rm H}}

\usepackage[stretch=50,shrink=50,kerning=true]{microtype}
\definecolor{tui_orange_dark}{cmyk}{0,0.6,1,0}
\definecolor{tui_orange_light}{cmyk}{0.0000,0.0876, 0.1474, 0.0157}
\definecolor{tui_green_dark}{cmyk}{1,0,0.5,0.2}
\definecolor{tui_green_light}{cmyk}{0.0576,0.0041, 0.0000, 0.0471}
\definecolor{tui_blue_dark}{cmyk}{1.0000,0.5000,0.0000,0.6000}
\definecolor{tui_blue_light}{cmyk}{0.0920,0.0440,0.0000,0.0196}
\definecolor{tui_red_dark}{cmyk}{0.0000,1.0000,1.0000,0.2000}
\definecolor{tui_red_light}{cmyk}{0.0000,0.1107,0.1107,0.0078}

\pgfplotscreateplotcyclelist{tui_dark}{
{thick,tui_orange_dark,mark=*,mark options={thin,scale=0.75, fill=tui_orange_dark!50!white}},
{thick,tui_green_dark,mark=square*,mark options={thin,scale=0.75, fill=tui_green_dark!50!white}},
{thick,tui_blue_dark,mark=triangle*,mark options={thin,scale=0.75, fill=tui_blue_dark!50!white}},
{thick,tui_red_dark,mark=diamond*,mark options={thin,scale=0.75, fill=tui_red_dark!50!white}},
}

\pgfplotsset{
    cycle list name=tui_dark,
}

\begin{document}
\abovedisplayskip      = 1pt plus 1pt minus 1pt
\abovedisplayshortskip = 1pt plus 1pt minus 1pt
\belowdisplayskip      = 1pt plus 1pt minus 1pt
\belowdisplayshortskip = 1pt plus 1pt minus 1pt
\title{Combining Matrix Design for 2D DoA Estimation with Compressive Antenna Arrays using Stochastic Gradient Descent}

%\author{\IEEEauthorblockN{Michael Shell}
%\IEEEauthorblockA{School of Electrical and\\Computer Engineering\\
%Georgia Institute of Technology\\
%Atlanta, Georgia 30332--0250\\
%Email: http://www.michaelshell.org/contact.html}
%\and
%\IEEEauthorblockN{Homer Simpson}
%\IEEEauthorblockA{Twentieth Century Fox\\
%\and
%\IEEEauthorblockN{James Kirk\\ and Montgomery Scott}
%\IEEEauthorblockA{Starfleet Academy\\
%San Francisco, California 96678--2391\\

\author{\IEEEauthorblockN{
S. Pawar\IEEEauthorrefmark{1}, % 1st author, 1st affiliations
S. Semper\IEEEauthorrefmark{1}, % 1st author, 1st affiliations
F. R\"omer\IEEEauthorrefmark{1}\IEEEauthorrefmark{2}
}                                     % ...
%\\
\IEEEauthorblockA{\IEEEauthorrefmark{1}% 1st affiliations
Institute for Information Technology, Technische Universit\"at Ilmenau, Ilmenau, Germany}
%\IEEEauthorblockA{\IEEEauthorrefmark{2}% 2nd affiliations
%School of Engineering, Pontifical Bolivarian University, Medellin-Bogota, Colombia}
%\IEEEauthorblockA{\IEEEauthorrefmark{4}% 4th affiliations
%(Affiliation): dept. name of organization, name of organization, acronyms acceptable, City, Country,
\IEEEauthorblockA{\IEEEauthorrefmark{2}% 1st affiliations
Fraunhofer Institute for Nondestructive Testing (IZFP), Ilmenau, Germany}
 Email: sankalp-prakash.pawar@tu-ilmenau.de
}

\maketitle

\begin{abstract}
%\boldmath

Recently, compressive antenna arrays have been considered for \gls{doa} estimation with reduced hardware complexity. By utilizing compressive sensing, such arrays employ a linear combining network to combine signals from a larger set of antenna elements in the analog RF domain. In this paper, we develop a design approach based on the minimization of error between spatial correlation function (SCF) of the compressive and the uncompressed array resulting in the estimation performance of the two arrays to be as close as possible. The proposed design is based on grid-free stochastic gradient descent (SGD) optimization. In addition to a low computational cost for the proposed method, we show numerically that the resulting combining matrices perform better than the ones generated by a previous approach and combining matrices generated from a Gaussian ensemble.
\end{abstract}

% \fi
\IEEEpeerreviewmaketitle

\section{Introduction}

Determination of the direction of impinging waves using an antenna array is formulated as the \gls{doa} estimation problem \cite{KV:96, HLVT:02}. It is well established that a comparatively large number of elements in the receiving array is required to achieve high \gls{doa} estimation accuracy \cite{HLVT:02}. Implementation of antenna arrays for \gls{doa} estimation usually requires a radio frequency (RF) chain to process each antenna output. Such a RF chain could include components such as a low-noise amplifier (LNA), filters, down-converter, and analog-to-digital (ADC) converter. Thus, realisation of a relatively large array brings along increased hardware costs.
In recent years, major focus of research has been dedicated to develop techniques that provide desirable \gls{doa} estimation performance while enabling reduction in hardware complexity. Ideas from the \gls{cs} domain have been considered to reduce the complexity of hardware implementation while maintaining \gls{doa} estimation accuracy \cite{HWKH:10, rossi2012spatial, shakeri2012direction,HL:14}.

One such approach involves the application of \gls{cs} paradigm in the spatial domain by employing an analog combining network to linearly combine $N$ antenna outputs to a smaller number of $M < N$ channels \cite{Wan09,GZS:11,semper181Ddoa}.  Using such a \gls{cs}-based combining network, only $M$ channels need to be processed through their respective RF chains. The hardware complexity of the \textit{compressive} array formed by the output channels of the combining network is lower compared to its uncompressed counterpart without such a network. Besides, the compressive array provides a better estimation performance owing to its larger aperture in comparison to an equivalent uncompressed array of size $M$.

In \cite{IbRaLaKo16}, a low-complexity design approach based on the \gls{scf} for 1D \gls{doa} estimation is proposed while its extension for 2D \Gls{doa} estimation is investigated in \cite{Lavrenko2018}. In \cite{Lavrenko2018}, instead of taking the complete 4D-\gls{scf} to define the cost function for optimization, only 2D subsets are used. A method to choose these 2D subsets, and an evaluation of the \gls{doa} performance while achieving considerable reduction in computational requirements compared to a direct extension of the approach in \cite{IbRaLaKo16} is described in \cite{Lavrenko2018}.

Despite the effectiveness of the approach in \cite{Lavrenko2018} its computational requirements increase substantially with increasing size of the antenna array. In this contribution, we propose a \gls{sgd} based approach with momentum \cite{robbins1951} for obtaining the kernels of the combining matrix. The cost function is defined as the average difference between the compressed and the uncompressed \Gls{scf} evaluated at a given set of 2D angles. The gradient of the cost function with respect to the combining matrix is analytically derived before applying the descent algorithm. We evaluate the performance of the proposed design approach for compression of a synthetic \gls{suca} in terms of the difference in the resulting spatial correlation functions as well as the \gls{crb}.

\glsreset{scf}
\glsreset{mse}
\section{System model}

\subsection{Input Signal}

For an $N$-element antenna array, the complex baseband signal received by it is given by $\bm{y}(t) \in \mathbb{C}^N$ such that,
\begin{equation}
    \bm{y}(t) = \sum_{s=1}^S \bm{a}(\theta_s, \vartheta_s) x_s(t) + \bm{v}(t), %= \bm{A} \bm{s}(t) + \bm{v}(t),
    \label{eq:model_conv}
\end{equation}
with the array receiving $S$ far-field narrowband plane waves impinging from \glspl{doa} described by azimuth ($\theta$) and elevation ($\vartheta$) pair $(\theta_s, \vartheta_s)$. The array steering vector is $\bm{a}(\theta_s, \vartheta_s)= [a_1(\theta_s, \vartheta_s), \ldots, a_N(\theta_s, \vartheta_s)]^{\rm H} \in \mathbb{C}^N$ for $s = 1, \dots, S$, $x_s(t)$ is the complex transmitted signal, while $\bm{v}(t) \in \mathbb{C}^N$ is the additive noise. Determination of angles pairs $(\theta_s, \vartheta_s)$ from $\bm{y}(t)$ is the goal of \gls{doa} estimation.

\subsection{Compressive Arrays}

A dedicated RF receiver chain is required for each antenna output of the array to comply with the model in \eqref{eq:model_conv}. Hardware implementation of such an array is prohibitive with considerations of complexity, cost, and power consumption. Reduction in the number of RF channels without the compromising \gls{doa} estimation performance can be achieved by applying the \textit{compressive} approach. A compressive array is obtained by utilizing an analog combining network to linearly combine the outputs of an array to a lower number of channels  \cite{IbRaLaKo16}.
The combining network is described by the matrix  $\bm{\Phi} \in \mathbb{C}^{M \times N}$, such that in baseband, $[\bm{\Phi}]_{m,n} = \alpha_{m,n}e^{\jmath \varphi_{m,n}}$ are the complex weights that are applied to the antenna outputs, with $m = 1, \dots M, n = 1, \dots, N$ and $M <N $. Thus, the array output of the compressive array is given by
\begin{align}
   \tilde{\bm{y}}(t)= \bm{\Phi}\bm{y}(t) =  \sum_{s=1}^S \tilde{\bm{a}}(\theta_s, \vartheta_s) x_s(t) + \bm{w}(t),
	\label{eqn:nm_doa}
\end{align}
where $\tilde{\bm{a}}(\theta, \vartheta) = \bm{\Phi} \bm{a} (\theta, \vartheta) \in \mathbb{C}^{M \times 1}$ represents the effective (compressed) array steering vector after the combining, and $\bm{w}(t) \in  \mathbb{C}^{M \times 1}$ is the noise vector representing additive noise sources of the system \cite{IbRaLaKo16}. Now, the goal is to design $\bm{\Phi}$ in order to allow accurate estimation of $(\theta_s, \vartheta_s)$ from the above compressive measurement $\tilde{\bm{y}}(t)$.

\section{Combining Matrix Design based on the Spatial Correlation Function}

For an array with a manifold $\bm{a}(\theta, \vartheta)$  the \gls{scf} $\rho : \mathbb{R}^4 \rightarrow \mathbb{C}$ is given by
\begin{equation}
    \rho(\theta_1, \theta_2, \vartheta_1, \vartheta_2) \overset{\Delta}{=} {\bm{a}}^{\rm H}(\theta_1, \vartheta_1){\bm{a}}(\theta_2, \vartheta_2)
    \label{eq:SCF}.
\end{equation}
Inserting  $\tilde{\bm{a}}(\theta, \vartheta)$ into \eqref{eq:SCF},  we obtain the \textit{effective} \gls{scf} of the compressive array as
\begin{align}
	\tilde{\rho}(\theta_1, \theta_2, \vartheta_1, \vartheta_2) &= \tilde{\bm{a}}^{\rm H}(\theta_1, \vartheta_1)\tilde{\bm{a}}(\theta_2, \vartheta_2) \notag \\ &= {\bm{a}}^{\rm H}(\theta_1, \vartheta_1) \bm{\Phi}^{\rm H} \bm{\Phi}{\bm{a}}(\theta_2, \vartheta_2)
	\notag \\ &= {\bm{a}}^{\rm H}(\theta_1, \vartheta_2) \bm{G}_{\rm \Phi} {\bm{a}}(\theta_1, \vartheta_2),
\end{align}
where for any arbitrary matrix $\bm{B}$ we denote by $\bm{G}_{\rm B} =  \bm{B}^{\rm H} \bm{B}$  its Gramian.

Depending on the application, one might have different requirements to the effective \gls{scf} $\tilde{\rho}(\theta_1, \theta_2, \vartheta_1, \vartheta_2)$, such that it ensures uniform sensitivity and/or good cross-correlation properties for instance. Our particular design goal in this work is to reduce the number of receiver channels $M$
while not compromising the DoA estimation performance compared to the original array without the combining. To evaluate the performance of a given $\bm{\Phi}$, we define $\delta : \mathbb{C}^{M \times N} \rightarrow \mathbb{R}$ as
\begin{equation}
    \delta(\bm{\Phi}) = \iint\displaylimits_{\theta}\iint\displaylimits_{\vartheta} |e(\bm{\Phi}, \theta_1, \theta_2, \vartheta_1, \vartheta_2)|^2 \mathrm{d} \theta_1 \mathrm{d} \theta_2 \mathrm{d} \vartheta_1 \mathrm{d} \vartheta_2,
    \label{eq:delta}
\end{equation}
where $e : \mathbb{C}^{M \times N} \times \mathbb{R}^4$ is defined as
\begin{align}
	e&(\bm{\Phi}, \theta_1, \theta_2, \vartheta_1, \vartheta_2) \overset{\Delta}{=} \tilde{\rho}(\theta_1, \theta_2, \vartheta_1, \vartheta_2)-{\rho}(\theta_1, \theta_2, \vartheta_1, \vartheta_2) \notag  \\
	&= {\bm{a}}^{\rm H}(\theta_1, \vartheta_1)\bm{G}_{\rm \Phi} {\bm{a}}(\theta_2, \vartheta_2) - {\bm{a}}^{\rm H}(\theta_1, \vartheta_1){\bm{a}}(\theta_2, \vartheta_2).	\label{eq:error_func}
\end{align}
Note that a small value of $\delta(\bm{\Phi})$ implies that we can expect a \gls{doa} estimation performance close to that of the original array before the combining, as discussed in \cite{IbRaLaKo16} for the case of 1D \gls{doa} estimation. So to this end $\delta$ serves as a suitable proxy to estimate the performance of a given $\bm\Phi$.

Consequently, in order to find a good combining matrix $\bm{\Phi}$ one has to minimize \eqref{eq:delta}. This however implies solving an optimization problem, where the evaluation of the objective function itself is very time consuming since it has to be approximated with some numerical integration scheme. To circumvent this problem, consider an i.i.d. sequence of random vectors $(\bm\Theta^k)_{k \in \mathbb{N}} \subset ([0, 2\pi) \times (0, \pi))^L$ over some appropriate probability measure space $(\Omega, \mathcal{A}, P)$ and $L \in \mathbb{N}$, which represent $L$ points on the unit sphere, where the first component is the angle in azimuth and the second elevation. Further, let us define
\[
    D(\bm\Phi, \bm\Theta) = \frac{1}{K L^2}\sum\limits_{k = 1}^K \sum\limits_{\ell_1}^{L} \sum\limits_{\ell_2}^{L} \vert e(\bm\Phi, \Theta^k_{\ell_1,1}, \Theta^k_{\ell_2,1}, \Theta^k_{\ell_1,2}, \Theta^k_{\ell_2,2})\vert^2.
\]
For the following, we set
\[
    \bm A_k = [\bm a(\Theta^k_{1,1}, \Theta^k_{1,2}), \dots, \bm a(\Theta^k_{L, 1}, \Theta^k_{L, 2})] \in \mathbb{C}^{M \times L}
\]
and
\[
    \bm E_k(\bm\Phi) = \bm{A}_k^{\rm H} \bm{G}_{\bm \Phi} \bm{A}_k - \bm{A}_k^{\rm H}\bm{A}_k \in \mathbb{C}^{L \times L},
\label{eq:error_mtx_full}
\]
which allows us to write
\[
D(\bm\Phi, \bm\Theta) = \frac{1}{K}\sum\limits_{k = 1}^K \Vert \bm E_k (\bm \Phi) \Vert^2_F\]
and
\[
\nabla_{\bm\Phi} D(\bm\Phi, \bm\Theta) = \frac{1}{K}\sum\limits_{k = 1}^K \nabla_{\bm\Phi} \Vert \bm E_k(\bm\Phi) \Vert^2_F.
\]

The key idea now is to minimize $\delta(\bm\Phi)$ by minimizing $D(\bm\Phi, \bm \Theta(\omega))$ for some fixed realization $\omega \in \Omega$. In other words we carry out the minimization of $D(\bm \Phi, \bm \Theta)$ for a specific realization of $\bm \Theta(\omega)$, which is a suitable approach, if we choose $K$ large enough and the distributions of the $\bm \Theta^k$ such that this random process explores the angular domain well enough. This stochastic approach is outlined in the following section.

\section{Stochastic Gradient Descent-based Design}
Gradient Descent with momentum \cite{LeCun2012} is a popular and simple first order technique to find local minima of smooth functions. In its most simple form it minimizes a smooth function $f : \mathbb{C}^{n} \rightarrow \mathbb{C}$ by iterating
\begin{align}
    \bm v_{i+1} &= \eta \bm v_i - \alpha \nabla_{\bm x}f(\bm x_i) \\
    \bm x_{i+1} &= \bm x_i - \bm v_i
    \label{SG_state_iter}
\end{align}
for initial velocity and state variables $\bm v_0, \bm x_0 \in \mathbb{C}^n$, a drag parameter $\eta \in [0,1)$ and a step size $\alpha > 0$. Now, if $f$ is of the form $f = (1/K)\sum_k^K g_k$ for smooth functions $g_k : \mathbb{C}^{n} \rightarrow \mathbb{C}$ above velocity can be rewritten as
\[
\bm v_{i+1} = \eta \bm v_i - \alpha \frac 1K \sum_k^K \nabla_{\bm x}g_k(\bm x_i).
\]
Now, suppose $K$ is prohibitively large such that the evaluation of $f$ and its gradient becomes computationally intractable. Then a well known approach that has recently gained a lot of traction because of its use in the machine learning community \cite{LeCun2012} is stochastic gradient descent \cite{robbins1951}, which exploits the fact that
\[
    \mathbb{E}\frac {1}{\kappa} \sum_{k \in \mathcal{K}} \nabla_{\bm x}g_k(\bm x_i) = \frac 1K \sum_k^K \nabla_{\bm x}g_k(\bm x_i)
\]
for $\mathcal{K}$ selected uniformly at random from the set of all subsets of $\{1, \dots, K\}$ with cardinality $\kappa$. This makes the above random sum depending on the random $\mathcal{K}$ a good approximation of $\nabla_{\bm x} f$ due to the weak law of large numbers, which states that with high probability above partial sum is close to its expectation for a single realization $\mathcal{K}(\omega)$. Then the velocity update reads as
\begin{align}
    \bm v_{i+1} = \eta \bm v_i - \alpha \frac 1\kappa \sum_{k \in \mathcal{K}} \nabla_{\bm x}g_k(\bm x_i),
    \label{SGD_iter}
\end{align}
while the state variable update stays the same as in \eqref{SG_state_iter}. So the proposed stochastic gradient descent consists of iteratively carrying out \eqref{SGD_iter} and \eqref{SG_state_iter}. Now, we outline how to apply the iterative approach explained above to the problem at hand.

\subsection{SGD for Combining Matrix Design}

For the specific problem of designing a suitable combining matrix $\bm\Phi$ we apply a slight modification of \eqref{SGD_iter} by reformulating it as an online minimization of $D(\bm\Phi, \bm\Theta)$ via
\begin{align}
    \bm v_{i+1} = \eta \bm v_i - \alpha \nabla_{\bm\Phi} \Vert \bm E_i(\bm\Phi_i) \Vert^2_F,
\end{align}
which allows to generate $\Theta^i$ and thus $\bm E_i(\bm\Phi)$ during the iteration without the need to store or precompute them during or before the execution. This means instead of randomly subselection points from a prespecified set of points on the sphere, we keep drawing a new set of points on the fly in each iteration step.

Clearly, we are still in need of an analytical expression of the gradient of $\Vert \bm E_i(\bm\Phi) \Vert^2_F$, which is provided in the following result.
\begin{lemma}
For given $\bm \Theta^k$ and $\bm\Phi \in \mathbb{C}^{M \times N}$ it holds that
\[
    \nabla_{\bm\Phi} \Vert \bm E_i(\bm\Phi_i) \Vert^2_F = 4 \bm\Phi \bm A_i \bm A_i^{\rm H} \bm G_{\bm \Phi} \bm A_i \bm A_i^{\rm H} - 4\bm \Phi \bm A_i \bm A_i^{\rm H}  \bm A_i \bm A_i^{\rm H}.
\]
\end{lemma}
\begin{proof}
First, we use the facts $\Vert \bm M \Vert^2_F = \Tr(\bm M^{\rm H} \bm M)$ and $\Tr(\bm A \bm B) = \Tr(\bm B \bm A)$ to get
\begin{align*}
\Vert \bm E_i(\bm\Phi) \Vert^2_F &=
    \Tr\left( \bm A_i^{\rm H} \bm\Phi^{\rm H} \bm\Phi \bm A_i \bm A_i^{\rm H} \bm \Phi^{\rm H} \bm \Phi \bm A_i \right) \\
    &-2\Tr\left(\bm A_i^{\rm H} \bm A_i \bm A_i^{\rm H} \bm \Phi^{\rm H} \bm\Phi \bm A_i\right) + \Tr\left(\bm A_i^{\rm H} \bm A_i \bm A_i^{\rm H} \bm A\right) \\
    &= \Tr\left( \bm A_i^{\rm H} \bm\Phi^{\rm H} \bm\Phi \bm A_i \bm A_i^{\rm H} \bm \Phi^{\rm H} \bm \Phi \bm A_i \right) \\
    &-2\Tr\left(\bm\Phi \bm A_i \bm A_i^{\rm H} \bm A_i \bm A_i^{\rm H} \bm \Phi^{\rm H} \right) + \Tr\left(\bm A_i^{\rm H} \bm A_i \bm A_i^{\rm H} \bm A\right).
\end{align*}
Now with two well known results from matrix calculus
\begin{align*}
\nabla_{\bm X} &\Tr\left( \bm M^{\rm H} \bm X^{\rm H} \bm X \bm M \bm M^{\rm H} \bm X^{\rm H} \bm X \bm M \right) = \\
    & 4 \bm X \bm M \bm M^{\rm H} \bm X^{\rm H} \bm X \bm M \bm M^{\rm H}
\end{align*}
and $\nabla_{\bm X} \Tr\left(\bm X \bm M \bm X^{\rm H} \right) = 2\bm X \bm M$,
we conclude the statement.
\end{proof}
Now we have all ingredients at hand to implement the proposed method with an analytically derived gradient calculation. The following section  evaluates the performance of this approach.

\section{Numerical Evaluation}
In this section, we assess the performance of the compressive array designed using our proposed \gls{sgd}-based approach in comparison with that of the uncompressed array, the compressive array obtained using the 2D \gls{scf}-based approach from \cite{Lavrenko2018} and the compressive array derived from a randomly drawn combining matrix.

In any case, we let \gls{sgd} run for $K = 5000$ steps, with $L = 250$ angles per step, step size $\alpha = 10^{-2}$ and drag parameter $\eta = 0.1$, where the distribution of the $\bm \Theta_k$ is the uniform distribution on $(0, 2\pi] \times [\pi / 4, 3\pi / 4]$. Moreover we always use normalized sensing matrices, which means that the columns of any $\bm \Phi$ considered are normalized to unit length. As an antenna we consider a \gls{suca} of $(\Sigma\times N_{S})$ isotropic elements, so it has $\Sigma = 3$ stacks of $N_{S} = 11$ elements each with the total number of elements denoted by $N = \Sigma \times N_{S}$. The array response of the \gls{suca} is given by $\bm a_{\mathrm{SUCA}}(\theta, \vartheta)$. We choose $d = 0.5 \lambda$ as the distance between two consecutive stacks, $R = 0.68 \lambda$ as the radius of the stacks, where $\lambda = {c}/{f}$ is the wave-length at frequency $f$ with $c = 3\cdot 10^8$~m/sec.
%
% We hereby denote the proposed \gls{sgd}-based approach as '\gls{sgd}-Opt', the 2D-\gls{scf}-based approach from \cite{Lavrenko2018} as '\gls{scf}-Opt', the randomization approach for combining matrix as 'Rand', and the uncompressed array as 'Uncomp'.
%
For the \gls{scf}-based approach from \cite{Lavrenko2018}, the number of grid points in azimuth and elevation used for calculation of \gls{scf} is $N_{\theta} = 121$ and $N_{\vartheta} = 61$, respectively over $\theta \in [-\pi, \pi]$ and  $\vartheta \in [-\pi/2, \pi/2]$.  Also, the number of reference points in elevation considered for optimization using this approach is $\vert\mathcal{N}\vert = 3$.

In Figure \ref{scf_error} we evaluate how the different approaches perform in terms of the overall \gls{scf}-error. To approximate the quantity in \eqref{eq:delta} we evaluate \eqref{eq:error_func} on a regular $2$D grid in azimuth and elevation and then sum the squared absolute values of these samples. This is done for different levels of compression rate $\rho$. Clearly, the proposed \gls{sgd}-based method is capable of approximating the original antenna response more closely than the previous approach and the conventional combining matrices resulting from a zero-mean Gaussian ensemble.

 % defined in \eqref{eq:delta}. Error function $e(\bm{\Phi}, \theta_1, \theta_2, \vartheta_1, \vartheta_2)$ from (\ref{eq:error_func}) can be represented by an $L \times L$ error matrix $\bm{E} = \bm{A}^{\rm H} \bm{G}_{\rm \Phi} \bm{A} - \bm{G}_{\rm A}$, where $\bm{A}$ is an $M \times L$ array manifold matrix containing array steering vectors $\bm{a}(\theta_{i}^{\rm (G)},\vartheta_{j}^{\rm (G)})$ for all $i = 1,2 \ldots, N_{\theta}$ and $j = 1, 2 \ldots, N_{\vartheta}$  with $L = N_{\theta}N_{\vartheta}$. In Normalized \gls{scf}-error is plotted against compression rate $\varrho = {M}/{N}$ in Fig. \ref{fig:scf_error}.

\begin{figure}
    \begin{tikzpicture}
    \begin{semilogyaxis}[
        height = 0.6\linewidth,
        width = \linewidth,
        grid=both,
        grid style={line width=.0pt, draw=none},
        major grid style={line width=.2pt,draw=gray!50},
        label style = {font=\small},
        ylabel = {SCF Distance},
        xlabel = {Compression Rate $\rho$},
        tick label style = {font=\tiny},
        legend style = {
            at={(0.5, -0.2)},
            anchor=north,
            legend columns=2,
            draw=none,
            fill=none
        },
        ]
    \addplot table[col sep=comma, x=rho, y=err0m] {figures/scf_error.csv};
    \addplot table[col sep=comma, x=rho, y=err1m] {figures/scf_error.csv};
    \addplot table[col sep=comma, x=rho, y=err2m] {figures/scf_error.csv};
    \legend{Random Gaussian, \gls{sgd}-based Approach, Previous Approach}
    \end{semilogyaxis}
    \end{tikzpicture}
    \vspace{-3mm}
    \caption{\small Distance of the SCFs generated by the different approaches for varying levels of compression.}\label{scf_error}
    \vspace{-6mm}
\end{figure}
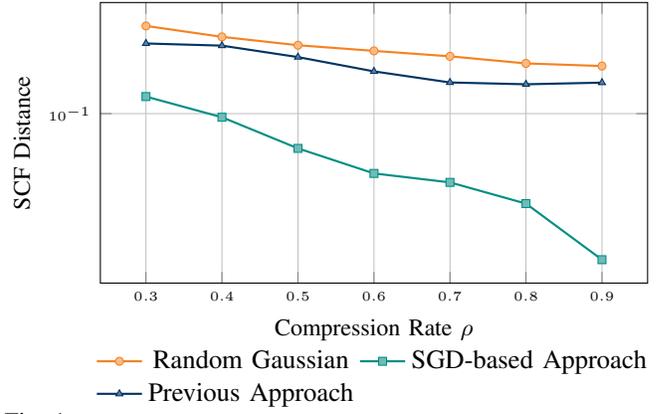

\begin{figure*}[!htb]
    \begin{center}
        \begin{tikzpicture}
            \node (origin) at (0,0) {};
            \begin{axis}[
    title=Random Gaussian,
    title style = {
        font=\bfseries,
        at = {(0.5, 0.9)},
        anchor = south
    },
    name=plot331,
    at=(origin),
    anchor=north,
    enlargelimits = false,
    axis on top = true,
    axis equal image,
    point meta min = -2.360772,
    point meta max = -1.318539,
    ,xtick=\empty,ylabel = {ele. [rad]},
    colormap={hot}{
        rgb(0.0 pt)=(0.0416,0.0,0.0),
        rgb(0.010101010101010102 pt)=(0.0621896881088697,0.0,0.0),
        rgb(0.020202020202020204 pt)=(0.09307422027217424,0.0,0.0),
        rgb(0.030303030303030304 pt)=(0.11366390838104393,0.0,0.0),
        rgb(0.04040404040404041 pt)=(0.14454844054434846,0.0,0.0),
        rgb(0.050505050505050504 pt)=(0.1651381286532182,0.0,0.0),
        rgb(0.06060606060606061 pt)=(0.1960226608165227,0.0,0.0),
        rgb(0.0707070707070707 pt)=(0.22690719297982725,0.0,0.0),
        rgb(0.08080808080808081 pt)=(0.24749688108869694,0.0,0.0),
        rgb(0.09090909090909091 pt)=(0.2783814132520015,0.0,0.0),
        rgb(0.10101010101010101 pt)=(0.2989711013608711,0.0,0.0),
        rgb(0.1111111111111111 pt)=(0.3298556335241757,0.0,0.0),
        rgb(0.12121212121212122 pt)=(0.36074016568748024,0.0,0.0),
        rgb(0.13131313131313133 pt)=(0.38132985379634987,0.0,0.0),
        rgb(0.1414141414141414 pt)=(0.41221438595965454,0.0,0.0),
        rgb(0.15151515151515152 pt)=(0.43280407406852417,0.0,0.0),
        rgb(0.16161616161616163 pt)=(0.4636886062318286,0.0,0.0),
        rgb(0.1717171717171717 pt)=(0.48427829434069847,0.0,0.0),
        rgb(0.18181818181818182 pt)=(0.5151628265040029,0.0,0.0),
        rgb(0.1919191919191919 pt)=(0.5460473586673074,0.0,0.0),
        rgb(0.20202020202020202 pt)=(0.5666370467761772,0.0,0.0),
        rgb(0.21212121212121213 pt)=(0.5975215789394817,0.0,0.0),
        rgb(0.2222222222222222 pt)=(0.6181112670483514,0.0,0.0),
        rgb(0.23232323232323232 pt)=(0.6489957992116561,0.0,0.0),
        rgb(0.24242424242424243 pt)=(0.6798803313749605,0.0,0.0),
        rgb(0.25252525252525254 pt)=(0.7004700194838303,0.0,0.0),
        rgb(0.26262626262626265 pt)=(0.7313545516471347,0.0,0.0),
        rgb(0.2727272727272727 pt)=(0.7519442397560044,0.0,0.0),
        rgb(0.2828282828282828 pt)=(0.782828771919309,0.0,0.0),
        rgb(0.29292929292929293 pt)=(0.8034184600281785,0.0,0.0),
        rgb(0.30303030303030304 pt)=(0.8343029921914833,0.0,0.0),
        rgb(0.31313131313131315 pt)=(0.8651875243547877,0.0,0.0),
        rgb(0.32323232323232326 pt)=(0.8857772124636573,0.0,0.0),
        rgb(0.3333333333333333 pt)=(0.9166617446269619,0.0,0.0),
        rgb(0.3434343434343434 pt)=(0.9372514327358317,0.0,0.0),
        rgb(0.35353535353535354 pt)=(0.9681359648991361,0.0,0.0),
        rgb(0.36363636363636365 pt)=(0.9990204970624408,0.0,0.0),
        rgb(0.37373737373737376 pt)=(1.0,0.019608769606337603,0.0),
        rgb(0.3838383838383838 pt)=(1.0,0.05049107236377195,0.0),
        rgb(0.3939393939393939 pt)=(1.0,0.07107927420206175,0.0),
        rgb(0.40404040404040403 pt)=(1.0,0.10196157695949624,0.0),
        rgb(0.41414141414141414 pt)=(1.0,0.1328438797169306,0.0),
        rgb(0.42424242424242425 pt)=(1.0,0.1534320815552204,0.0),
        rgb(0.43434343434343436 pt)=(1.0,0.1843143843126549,0.0),
        rgb(0.4444444444444444 pt)=(1.0,0.20490258615094456,0.0),
        rgb(0.45454545454545453 pt)=(1.0,0.23578488890837906,0.0),
        rgb(0.46464646464646464 pt)=(1.0,0.2563730907466687,0.0),
        rgb(0.47474747474747475 pt)=(1.0,0.28725539350410323,0.0),
        rgb(0.48484848484848486 pt)=(1.0,0.3181376962615377,0.0),
        rgb(0.494949494949495 pt)=(1.0,0.33872589809982734,0.0),
        rgb(0.5050505050505051 pt)=(1.0,0.36960820085726187,0.0),
        rgb(0.5151515151515151 pt)=(1.0,0.3901964026955515,0.0),
        rgb(0.5252525252525253 pt)=(1.0,0.421078705452986,0.0),
        rgb(0.5353535353535354 pt)=(1.0,0.4519610082104205,0.0),
        rgb(0.5454545454545454 pt)=(1.0,0.47254921004871014,0.0),
        rgb(0.5555555555555556 pt)=(1.0,0.5034315128061446,0.0),
        rgb(0.5656565656565656 pt)=(1.0,0.5240197146444343,0.0),
        rgb(0.5757575757575758 pt)=(1.0,0.5549020174018688,0.0),
        rgb(0.5858585858585859 pt)=(1.0,0.5754902192401584,0.0),
        rgb(0.5959595959595959 pt)=(1.0,0.606372521997593,0.0),
        rgb(0.6060606060606061 pt)=(1.0,0.6372548247550275,0.0),
        rgb(0.6161616161616161 pt)=(1.0,0.6578430265933172,0.0),
        rgb(0.6262626262626263 pt)=(1.0,0.6887253293507516,0.0),
        rgb(0.6363636363636364 pt)=(1.0,0.7093135311890413,0.0),
        rgb(0.6464646464646465 pt)=(1.0,0.7401958339464758,0.0),
        rgb(0.6565656565656566 pt)=(1.0,0.7710781367039102,0.0),
        rgb(0.6666666666666666 pt)=(1.0,0.7916663385421999,0.0),
        rgb(0.6767676767676768 pt)=(1.0,0.8225486412996345,0.0),
        rgb(0.6868686868686869 pt)=(1.0,0.843136843137924,0.0),
        rgb(0.696969696969697 pt)=(1.0,0.8740191458953586,0.0),
        rgb(0.7070707070707071 pt)=(1.0,0.9049014486527931,0.0),
        rgb(0.7171717171717171 pt)=(1.0,0.9254896504910827,0.0),
        rgb(0.7272727272727273 pt)=(1.0,0.9563719532485172,0.0),
        rgb(0.7373737373737373 pt)=(1.0,0.9769601550868069,0.0),
        rgb(0.7474747474747475 pt)=(1.0,1.0,0.011763717646070686),
        rgb(0.7575757575757576 pt)=(1.0,1.0,0.04264610146963098),
        rgb(0.7676767676767676 pt)=(1.0,1.0,0.08896967720497098),
        rgb(0.7777777777777778 pt)=(1.0,1.0,0.13529325294031186),
        rgb(0.7878787878787878 pt)=(1.0,1.0,0.16617563676387215),
        rgb(0.797979797979798 pt)=(1.0,1.0,0.21249921249921258),
        rgb(0.8080808080808081 pt)=(1.0,1.0,0.24338159632277287),
        rgb(0.8181818181818182 pt)=(1.0,1.0,0.2897051720581133),
        rgb(0.8282828282828283 pt)=(1.0,1.0,0.3360287477934533),
        rgb(0.8383838383838383 pt)=(1.0,1.0,0.36691113161701405),
        rgb(0.8484848484848485 pt)=(1.0,1.0,0.41323470735235446),
        rgb(0.8585858585858586 pt)=(1.0,1.0,0.44411709117591475),
        rgb(0.8686868686868687 pt)=(1.0,1.0,0.4904406669112552),
        rgb(0.8787878787878788 pt)=(1.0,1.0,0.5213230507348154),
        rgb(0.8888888888888888 pt)=(1.0,1.0,0.567646626470156),
        rgb(0.898989898989899 pt)=(1.0,1.0,0.6139702022054964),
        rgb(0.9090909090909091 pt)=(1.0,1.0,0.6448525860290567),
        rgb(0.9191919191919192 pt)=(1.0,1.0,0.6911761617643971),
        rgb(0.9292929292929293 pt)=(1.0,1.0,0.7220585455879573),
        rgb(0.9393939393939394 pt)=(1.0,1.0,0.7683821213232979),
        rgb(0.9494949494949495 pt)=(1.0,1.0,0.8147056970586383),
        rgb(0.9595959595959596 pt)=(1.0,1.0,0.8455880808821985),
        rgb(0.9696969696969697 pt)=(1.0,1.0,0.891911656617539),
        rgb(0.9797979797979798 pt)=(1.0,1.0,0.9227940404410993),
        rgb(0.98989898989899 pt)=(1.0,1.0,0.9691176161764397),
        rgb(1.0 pt)=(1.0,1.0,1.0),},
    label style = {font=\small},
    xticklabel style = {font=\tiny},
    colorbar style ={
        tick label style = {font=\tiny},
        xticklabel={$10^{
        \pgfmathparse{\tick}
        \pgfmathprintnumber[precision=2]{\pgfmathresult}}$},
        at = {(1.03,1)},
        anchor = north west,
    }
    ]
    \addplot graphics [
        xmin = -3.141593,
        xmax = 3.141593,
        ymin = -0.785398,
        ymax = 0.785398
    ] {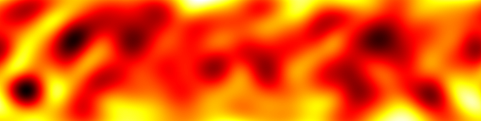};
\end{axis}
            \begin{axis}[
    title=,
    title style = {
        font=\bfseries,
        at = {(0.5, 0.9)},
        anchor = south
    },
    name=plot332,
    at=(plot331.below south),
    anchor=north,
    enlargelimits = false,
    axis on top = true,
    axis equal image,
    point meta min = -0.344911,
    point meta max = 2.098015,
    ,xtick=\empty,ylabel = {ele. [rad]},
    colormap={hot}{
        rgb(0.0 pt)=(0.0416,0.0,0.0),
        rgb(0.010101010101010102 pt)=(0.0621896881088697,0.0,0.0),
        rgb(0.020202020202020204 pt)=(0.09307422027217424,0.0,0.0),
        rgb(0.030303030303030304 pt)=(0.11366390838104393,0.0,0.0),
        rgb(0.04040404040404041 pt)=(0.14454844054434846,0.0,0.0),
        rgb(0.050505050505050504 pt)=(0.1651381286532182,0.0,0.0),
        rgb(0.06060606060606061 pt)=(0.1960226608165227,0.0,0.0),
        rgb(0.0707070707070707 pt)=(0.22690719297982725,0.0,0.0),
        rgb(0.08080808080808081 pt)=(0.24749688108869694,0.0,0.0),
        rgb(0.09090909090909091 pt)=(0.2783814132520015,0.0,0.0),
        rgb(0.10101010101010101 pt)=(0.2989711013608711,0.0,0.0),
        rgb(0.1111111111111111 pt)=(0.3298556335241757,0.0,0.0),
        rgb(0.12121212121212122 pt)=(0.36074016568748024,0.0,0.0),
        rgb(0.13131313131313133 pt)=(0.38132985379634987,0.0,0.0),
        rgb(0.1414141414141414 pt)=(0.41221438595965454,0.0,0.0),
        rgb(0.15151515151515152 pt)=(0.43280407406852417,0.0,0.0),
        rgb(0.16161616161616163 pt)=(0.4636886062318286,0.0,0.0),
        rgb(0.1717171717171717 pt)=(0.48427829434069847,0.0,0.0),
        rgb(0.18181818181818182 pt)=(0.5151628265040029,0.0,0.0),
        rgb(0.1919191919191919 pt)=(0.5460473586673074,0.0,0.0),
        rgb(0.20202020202020202 pt)=(0.5666370467761772,0.0,0.0),
        rgb(0.21212121212121213 pt)=(0.5975215789394817,0.0,0.0),
        rgb(0.2222222222222222 pt)=(0.6181112670483514,0.0,0.0),
        rgb(0.23232323232323232 pt)=(0.6489957992116561,0.0,0.0),
        rgb(0.24242424242424243 pt)=(0.6798803313749605,0.0,0.0),
        rgb(0.25252525252525254 pt)=(0.7004700194838303,0.0,0.0),
        rgb(0.26262626262626265 pt)=(0.7313545516471347,0.0,0.0),
        rgb(0.2727272727272727 pt)=(0.7519442397560044,0.0,0.0),
        rgb(0.2828282828282828 pt)=(0.782828771919309,0.0,0.0),
        rgb(0.29292929292929293 pt)=(0.8034184600281785,0.0,0.0),
        rgb(0.30303030303030304 pt)=(0.8343029921914833,0.0,0.0),
        rgb(0.31313131313131315 pt)=(0.8651875243547877,0.0,0.0),
        rgb(0.32323232323232326 pt)=(0.8857772124636573,0.0,0.0),
        rgb(0.3333333333333333 pt)=(0.9166617446269619,0.0,0.0),
        rgb(0.3434343434343434 pt)=(0.9372514327358317,0.0,0.0),
        rgb(0.35353535353535354 pt)=(0.9681359648991361,0.0,0.0),
        rgb(0.36363636363636365 pt)=(0.9990204970624408,0.0,0.0),
        rgb(0.37373737373737376 pt)=(1.0,0.019608769606337603,0.0),
        rgb(0.3838383838383838 pt)=(1.0,0.05049107236377195,0.0),
        rgb(0.3939393939393939 pt)=(1.0,0.07107927420206175,0.0),
        rgb(0.40404040404040403 pt)=(1.0,0.10196157695949624,0.0),
        rgb(0.41414141414141414 pt)=(1.0,0.1328438797169306,0.0),
        rgb(0.42424242424242425 pt)=(1.0,0.1534320815552204,0.0),
        rgb(0.43434343434343436 pt)=(1.0,0.1843143843126549,0.0),
        rgb(0.4444444444444444 pt)=(1.0,0.20490258615094456,0.0),
        rgb(0.45454545454545453 pt)=(1.0,0.23578488890837906,0.0),
        rgb(0.46464646464646464 pt)=(1.0,0.2563730907466687,0.0),
        rgb(0.47474747474747475 pt)=(1.0,0.28725539350410323,0.0),
        rgb(0.48484848484848486 pt)=(1.0,0.3181376962615377,0.0),
        rgb(0.494949494949495 pt)=(1.0,0.33872589809982734,0.0),
        rgb(0.5050505050505051 pt)=(1.0,0.36960820085726187,0.0),
        rgb(0.5151515151515151 pt)=(1.0,0.3901964026955515,0.0),
        rgb(0.5252525252525253 pt)=(1.0,0.421078705452986,0.0),
        rgb(0.5353535353535354 pt)=(1.0,0.4519610082104205,0.0),
        rgb(0.5454545454545454 pt)=(1.0,0.47254921004871014,0.0),
        rgb(0.5555555555555556 pt)=(1.0,0.5034315128061446,0.0),
        rgb(0.5656565656565656 pt)=(1.0,0.5240197146444343,0.0),
        rgb(0.5757575757575758 pt)=(1.0,0.5549020174018688,0.0),
        rgb(0.5858585858585859 pt)=(1.0,0.5754902192401584,0.0),
        rgb(0.5959595959595959 pt)=(1.0,0.606372521997593,0.0),
        rgb(0.6060606060606061 pt)=(1.0,0.6372548247550275,0.0),
        rgb(0.6161616161616161 pt)=(1.0,0.6578430265933172,0.0),
        rgb(0.6262626262626263 pt)=(1.0,0.6887253293507516,0.0),
        rgb(0.6363636363636364 pt)=(1.0,0.7093135311890413,0.0),
        rgb(0.6464646464646465 pt)=(1.0,0.7401958339464758,0.0),
        rgb(0.6565656565656566 pt)=(1.0,0.7710781367039102,0.0),
        rgb(0.6666666666666666 pt)=(1.0,0.7916663385421999,0.0),
        rgb(0.6767676767676768 pt)=(1.0,0.8225486412996345,0.0),
        rgb(0.6868686868686869 pt)=(1.0,0.843136843137924,0.0),
        rgb(0.696969696969697 pt)=(1.0,0.8740191458953586,0.0),
        rgb(0.7070707070707071 pt)=(1.0,0.9049014486527931,0.0),
        rgb(0.7171717171717171 pt)=(1.0,0.9254896504910827,0.0),
        rgb(0.7272727272727273 pt)=(1.0,0.9563719532485172,0.0),
        rgb(0.7373737373737373 pt)=(1.0,0.9769601550868069,0.0),
        rgb(0.7474747474747475 pt)=(1.0,1.0,0.011763717646070686),
        rgb(0.7575757575757576 pt)=(1.0,1.0,0.04264610146963098),
        rgb(0.7676767676767676 pt)=(1.0,1.0,0.08896967720497098),
        rgb(0.7777777777777778 pt)=(1.0,1.0,0.13529325294031186),
        rgb(0.7878787878787878 pt)=(1.0,1.0,0.16617563676387215),
        rgb(0.797979797979798 pt)=(1.0,1.0,0.21249921249921258),
        rgb(0.8080808080808081 pt)=(1.0,1.0,0.24338159632277287),
        rgb(0.8181818181818182 pt)=(1.0,1.0,0.2897051720581133),
        rgb(0.8282828282828283 pt)=(1.0,1.0,0.3360287477934533),
        rgb(0.8383838383838383 pt)=(1.0,1.0,0.36691113161701405),
        rgb(0.8484848484848485 pt)=(1.0,1.0,0.41323470735235446),
        rgb(0.8585858585858586 pt)=(1.0,1.0,0.44411709117591475),
        rgb(0.8686868686868687 pt)=(1.0,1.0,0.4904406669112552),
        rgb(0.8787878787878788 pt)=(1.0,1.0,0.5213230507348154),
        rgb(0.8888888888888888 pt)=(1.0,1.0,0.567646626470156),
        rgb(0.898989898989899 pt)=(1.0,1.0,0.6139702022054964),
        rgb(0.9090909090909091 pt)=(1.0,1.0,0.6448525860290567),
        rgb(0.9191919191919192 pt)=(1.0,1.0,0.6911761617643971),
        rgb(0.9292929292929293 pt)=(1.0,1.0,0.7220585455879573),
        rgb(0.9393939393939394 pt)=(1.0,1.0,0.7683821213232979),
        rgb(0.9494949494949495 pt)=(1.0,1.0,0.8147056970586383),
        rgb(0.9595959595959596 pt)=(1.0,1.0,0.8455880808821985),
        rgb(0.9696969696969697 pt)=(1.0,1.0,0.891911656617539),
        rgb(0.9797979797979798 pt)=(1.0,1.0,0.9227940404410993),
        rgb(0.98989898989899 pt)=(1.0,1.0,0.9691176161764397),
        rgb(1.0 pt)=(1.0,1.0,1.0),},
    label style = {font=\small},
    xticklabel style = {font=\tiny},
    colorbar style ={
        tick label style = {font=\tiny},
        xticklabel={$10^{
        \pgfmathparse{\tick}
        \pgfmathprintnumber[precision=2]{\pgfmathresult}}$},
        at = {(1.03,1)},
        anchor = north west,
    }
    ]
    \addplot graphics [
        xmin = -3.141593,
        xmax = 3.141593,
        ymin = -0.785398,
        ymax = 0.785398
    ] {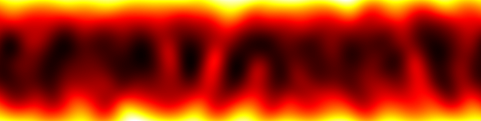};
\end{axis}
            \begin{axis}[
    title=,
    title style = {
        font=\bfseries,
        at = {(0.5, 0.9)},
        anchor = south
    },
    name=plot333,
    at=(plot332.below south),
    anchor=north,
    enlargelimits = false,
    axis on top = true,
    axis equal image,
    point meta min = -1.215294,
    point meta max = 1.407569,
    ,xlabel = {az. [rad]},ylabel = {ele. [rad]},
    colormap={hot}{
        rgb(0.0 pt)=(0.0416,0.0,0.0),
        rgb(0.010101010101010102 pt)=(0.0621896881088697,0.0,0.0),
        rgb(0.020202020202020204 pt)=(0.09307422027217424,0.0,0.0),
        rgb(0.030303030303030304 pt)=(0.11366390838104393,0.0,0.0),
        rgb(0.04040404040404041 pt)=(0.14454844054434846,0.0,0.0),
        rgb(0.050505050505050504 pt)=(0.1651381286532182,0.0,0.0),
        rgb(0.06060606060606061 pt)=(0.1960226608165227,0.0,0.0),
        rgb(0.0707070707070707 pt)=(0.22690719297982725,0.0,0.0),
        rgb(0.08080808080808081 pt)=(0.24749688108869694,0.0,0.0),
        rgb(0.09090909090909091 pt)=(0.2783814132520015,0.0,0.0),
        rgb(0.10101010101010101 pt)=(0.2989711013608711,0.0,0.0),
        rgb(0.1111111111111111 pt)=(0.3298556335241757,0.0,0.0),
        rgb(0.12121212121212122 pt)=(0.36074016568748024,0.0,0.0),
        rgb(0.13131313131313133 pt)=(0.38132985379634987,0.0,0.0),
        rgb(0.1414141414141414 pt)=(0.41221438595965454,0.0,0.0),
        rgb(0.15151515151515152 pt)=(0.43280407406852417,0.0,0.0),
        rgb(0.16161616161616163 pt)=(0.4636886062318286,0.0,0.0),
        rgb(0.1717171717171717 pt)=(0.48427829434069847,0.0,0.0),
        rgb(0.18181818181818182 pt)=(0.5151628265040029,0.0,0.0),
        rgb(0.1919191919191919 pt)=(0.5460473586673074,0.0,0.0),
        rgb(0.20202020202020202 pt)=(0.5666370467761772,0.0,0.0),
        rgb(0.21212121212121213 pt)=(0.5975215789394817,0.0,0.0),
        rgb(0.2222222222222222 pt)=(0.6181112670483514,0.0,0.0),
        rgb(0.23232323232323232 pt)=(0.6489957992116561,0.0,0.0),
        rgb(0.24242424242424243 pt)=(0.6798803313749605,0.0,0.0),
        rgb(0.25252525252525254 pt)=(0.7004700194838303,0.0,0.0),
        rgb(0.26262626262626265 pt)=(0.7313545516471347,0.0,0.0),
        rgb(0.2727272727272727 pt)=(0.7519442397560044,0.0,0.0),
        rgb(0.2828282828282828 pt)=(0.782828771919309,0.0,0.0),
        rgb(0.29292929292929293 pt)=(0.8034184600281785,0.0,0.0),
        rgb(0.30303030303030304 pt)=(0.8343029921914833,0.0,0.0),
        rgb(0.31313131313131315 pt)=(0.8651875243547877,0.0,0.0),
        rgb(0.32323232323232326 pt)=(0.8857772124636573,0.0,0.0),
        rgb(0.3333333333333333 pt)=(0.9166617446269619,0.0,0.0),
        rgb(0.3434343434343434 pt)=(0.9372514327358317,0.0,0.0),
        rgb(0.35353535353535354 pt)=(0.9681359648991361,0.0,0.0),
        rgb(0.36363636363636365 pt)=(0.9990204970624408,0.0,0.0),
        rgb(0.37373737373737376 pt)=(1.0,0.019608769606337603,0.0),
        rgb(0.3838383838383838 pt)=(1.0,0.05049107236377195,0.0),
        rgb(0.3939393939393939 pt)=(1.0,0.07107927420206175,0.0),
        rgb(0.40404040404040403 pt)=(1.0,0.10196157695949624,0.0),
        rgb(0.41414141414141414 pt)=(1.0,0.1328438797169306,0.0),
        rgb(0.42424242424242425 pt)=(1.0,0.1534320815552204,0.0),
        rgb(0.43434343434343436 pt)=(1.0,0.1843143843126549,0.0),
        rgb(0.4444444444444444 pt)=(1.0,0.20490258615094456,0.0),
        rgb(0.45454545454545453 pt)=(1.0,0.23578488890837906,0.0),
        rgb(0.46464646464646464 pt)=(1.0,0.2563730907466687,0.0),
        rgb(0.47474747474747475 pt)=(1.0,0.28725539350410323,0.0),
        rgb(0.48484848484848486 pt)=(1.0,0.3181376962615377,0.0),
        rgb(0.494949494949495 pt)=(1.0,0.33872589809982734,0.0),
        rgb(0.5050505050505051 pt)=(1.0,0.36960820085726187,0.0),
        rgb(0.5151515151515151 pt)=(1.0,0.3901964026955515,0.0),
        rgb(0.5252525252525253 pt)=(1.0,0.421078705452986,0.0),
        rgb(0.5353535353535354 pt)=(1.0,0.4519610082104205,0.0),
        rgb(0.5454545454545454 pt)=(1.0,0.47254921004871014,0.0),
        rgb(0.5555555555555556 pt)=(1.0,0.5034315128061446,0.0),
        rgb(0.5656565656565656 pt)=(1.0,0.5240197146444343,0.0),
        rgb(0.5757575757575758 pt)=(1.0,0.5549020174018688,0.0),
        rgb(0.5858585858585859 pt)=(1.0,0.5754902192401584,0.0),
        rgb(0.5959595959595959 pt)=(1.0,0.606372521997593,0.0),
        rgb(0.6060606060606061 pt)=(1.0,0.6372548247550275,0.0),
        rgb(0.6161616161616161 pt)=(1.0,0.6578430265933172,0.0),
        rgb(0.6262626262626263 pt)=(1.0,0.6887253293507516,0.0),
        rgb(0.6363636363636364 pt)=(1.0,0.7093135311890413,0.0),
        rgb(0.6464646464646465 pt)=(1.0,0.7401958339464758,0.0),
        rgb(0.6565656565656566 pt)=(1.0,0.7710781367039102,0.0),
        rgb(0.6666666666666666 pt)=(1.0,0.7916663385421999,0.0),
        rgb(0.6767676767676768 pt)=(1.0,0.8225486412996345,0.0),
        rgb(0.6868686868686869 pt)=(1.0,0.843136843137924,0.0),
        rgb(0.696969696969697 pt)=(1.0,0.8740191458953586,0.0),
        rgb(0.7070707070707071 pt)=(1.0,0.9049014486527931,0.0),
        rgb(0.7171717171717171 pt)=(1.0,0.9254896504910827,0.0),
        rgb(0.7272727272727273 pt)=(1.0,0.9563719532485172,0.0),
        rgb(0.7373737373737373 pt)=(1.0,0.9769601550868069,0.0),
        rgb(0.7474747474747475 pt)=(1.0,1.0,0.011763717646070686),
        rgb(0.7575757575757576 pt)=(1.0,1.0,0.04264610146963098),
        rgb(0.7676767676767676 pt)=(1.0,1.0,0.08896967720497098),
        rgb(0.7777777777777778 pt)=(1.0,1.0,0.13529325294031186),
        rgb(0.7878787878787878 pt)=(1.0,1.0,0.16617563676387215),
        rgb(0.797979797979798 pt)=(1.0,1.0,0.21249921249921258),
        rgb(0.8080808080808081 pt)=(1.0,1.0,0.24338159632277287),
        rgb(0.8181818181818182 pt)=(1.0,1.0,0.2897051720581133),
        rgb(0.8282828282828283 pt)=(1.0,1.0,0.3360287477934533),
        rgb(0.8383838383838383 pt)=(1.0,1.0,0.36691113161701405),
        rgb(0.8484848484848485 pt)=(1.0,1.0,0.41323470735235446),
        rgb(0.8585858585858586 pt)=(1.0,1.0,0.44411709117591475),
        rgb(0.8686868686868687 pt)=(1.0,1.0,0.4904406669112552),
        rgb(0.8787878787878788 pt)=(1.0,1.0,0.5213230507348154),
        rgb(0.8888888888888888 pt)=(1.0,1.0,0.567646626470156),
        rgb(0.898989898989899 pt)=(1.0,1.0,0.6139702022054964),
        rgb(0.9090909090909091 pt)=(1.0,1.0,0.6448525860290567),
        rgb(0.9191919191919192 pt)=(1.0,1.0,0.6911761617643971),
        rgb(0.9292929292929293 pt)=(1.0,1.0,0.7220585455879573),
        rgb(0.9393939393939394 pt)=(1.0,1.0,0.7683821213232979),
        rgb(0.9494949494949495 pt)=(1.0,1.0,0.8147056970586383),
        rgb(0.9595959595959596 pt)=(1.0,1.0,0.8455880808821985),
        rgb(0.9696969696969697 pt)=(1.0,1.0,0.891911656617539),
        rgb(0.9797979797979798 pt)=(1.0,1.0,0.9227940404410993),
        rgb(0.98989898989899 pt)=(1.0,1.0,0.9691176161764397),
        rgb(1.0 pt)=(1.0,1.0,1.0),},
    label style = {font=\small},
    xticklabel style = {font=\tiny},
    colorbar style ={
        tick label style = {font=\tiny},
        xticklabel={$10^{
        \pgfmathparse{\tick}
        \pgfmathprintnumber[precision=2]{\pgfmathresult}}$},
        at = {(1.03,1)},
        anchor = north west,
    }
    ]
    \addplot graphics [
        xmin = -3.141593,
        xmax = 3.141593,
        ymin = -0.785398,
        ymax = 0.785398
    ] {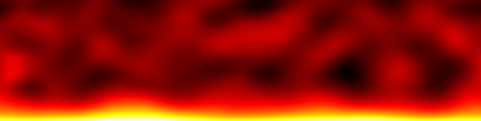};
\end{axis}
            \begin{axis}[
    title=Previous Approach,
    title style = {
        font=\bfseries,
        at = {(0.5, 0.9)},
        anchor = south
    },
    name=plot334,
    at=(plot331.right of east),
    anchor=left of west,
    enlargelimits = false,
    axis on top = true,
    axis equal image,
    point meta min = -2.360772,
    point meta max = -1.318539,
    colorbar,xtick=\empty,ytick=\empty,
    colormap={hot}{
        rgb(0.0 pt)=(0.0416,0.0,0.0),
        rgb(0.010101010101010102 pt)=(0.0621896881088697,0.0,0.0),
        rgb(0.020202020202020204 pt)=(0.09307422027217424,0.0,0.0),
        rgb(0.030303030303030304 pt)=(0.11366390838104393,0.0,0.0),
        rgb(0.04040404040404041 pt)=(0.14454844054434846,0.0,0.0),
        rgb(0.050505050505050504 pt)=(0.1651381286532182,0.0,0.0),
        rgb(0.06060606060606061 pt)=(0.1960226608165227,0.0,0.0),
        rgb(0.0707070707070707 pt)=(0.22690719297982725,0.0,0.0),
        rgb(0.08080808080808081 pt)=(0.24749688108869694,0.0,0.0),
        rgb(0.09090909090909091 pt)=(0.2783814132520015,0.0,0.0),
        rgb(0.10101010101010101 pt)=(0.2989711013608711,0.0,0.0),
        rgb(0.1111111111111111 pt)=(0.3298556335241757,0.0,0.0),
        rgb(0.12121212121212122 pt)=(0.36074016568748024,0.0,0.0),
        rgb(0.13131313131313133 pt)=(0.38132985379634987,0.0,0.0),
        rgb(0.1414141414141414 pt)=(0.41221438595965454,0.0,0.0),
        rgb(0.15151515151515152 pt)=(0.43280407406852417,0.0,0.0),
        rgb(0.16161616161616163 pt)=(0.4636886062318286,0.0,0.0),
        rgb(0.1717171717171717 pt)=(0.48427829434069847,0.0,0.0),
        rgb(0.18181818181818182 pt)=(0.5151628265040029,0.0,0.0),
        rgb(0.1919191919191919 pt)=(0.5460473586673074,0.0,0.0),
        rgb(0.20202020202020202 pt)=(0.5666370467761772,0.0,0.0),
        rgb(0.21212121212121213 pt)=(0.5975215789394817,0.0,0.0),
        rgb(0.2222222222222222 pt)=(0.6181112670483514,0.0,0.0),
        rgb(0.23232323232323232 pt)=(0.6489957992116561,0.0,0.0),
        rgb(0.24242424242424243 pt)=(0.6798803313749605,0.0,0.0),
        rgb(0.25252525252525254 pt)=(0.7004700194838303,0.0,0.0),
        rgb(0.26262626262626265 pt)=(0.7313545516471347,0.0,0.0),
        rgb(0.2727272727272727 pt)=(0.7519442397560044,0.0,0.0),
        rgb(0.2828282828282828 pt)=(0.782828771919309,0.0,0.0),
        rgb(0.29292929292929293 pt)=(0.8034184600281785,0.0,0.0),
        rgb(0.30303030303030304 pt)=(0.8343029921914833,0.0,0.0),
        rgb(0.31313131313131315 pt)=(0.8651875243547877,0.0,0.0),
        rgb(0.32323232323232326 pt)=(0.8857772124636573,0.0,0.0),
        rgb(0.3333333333333333 pt)=(0.9166617446269619,0.0,0.0),
        rgb(0.3434343434343434 pt)=(0.9372514327358317,0.0,0.0),
        rgb(0.35353535353535354 pt)=(0.9681359648991361,0.0,0.0),
        rgb(0.36363636363636365 pt)=(0.9990204970624408,0.0,0.0),
        rgb(0.37373737373737376 pt)=(1.0,0.019608769606337603,0.0),
        rgb(0.3838383838383838 pt)=(1.0,0.05049107236377195,0.0),
        rgb(0.3939393939393939 pt)=(1.0,0.07107927420206175,0.0),
        rgb(0.40404040404040403 pt)=(1.0,0.10196157695949624,0.0),
        rgb(0.41414141414141414 pt)=(1.0,0.1328438797169306,0.0),
        rgb(0.42424242424242425 pt)=(1.0,0.1534320815552204,0.0),
        rgb(0.43434343434343436 pt)=(1.0,0.1843143843126549,0.0),
        rgb(0.4444444444444444 pt)=(1.0,0.20490258615094456,0.0),
        rgb(0.45454545454545453 pt)=(1.0,0.23578488890837906,0.0),
        rgb(0.46464646464646464 pt)=(1.0,0.2563730907466687,0.0),
        rgb(0.47474747474747475 pt)=(1.0,0.28725539350410323,0.0),
        rgb(0.48484848484848486 pt)=(1.0,0.3181376962615377,0.0),
        rgb(0.494949494949495 pt)=(1.0,0.33872589809982734,0.0),
        rgb(0.5050505050505051 pt)=(1.0,0.36960820085726187,0.0),
        rgb(0.5151515151515151 pt)=(1.0,0.3901964026955515,0.0),
        rgb(0.5252525252525253 pt)=(1.0,0.421078705452986,0.0),
        rgb(0.5353535353535354 pt)=(1.0,0.4519610082104205,0.0),
        rgb(0.5454545454545454 pt)=(1.0,0.47254921004871014,0.0),
        rgb(0.5555555555555556 pt)=(1.0,0.5034315128061446,0.0),
        rgb(0.5656565656565656 pt)=(1.0,0.5240197146444343,0.0),
        rgb(0.5757575757575758 pt)=(1.0,0.5549020174018688,0.0),
        rgb(0.5858585858585859 pt)=(1.0,0.5754902192401584,0.0),
        rgb(0.5959595959595959 pt)=(1.0,0.606372521997593,0.0),
        rgb(0.6060606060606061 pt)=(1.0,0.6372548247550275,0.0),
        rgb(0.6161616161616161 pt)=(1.0,0.6578430265933172,0.0),
        rgb(0.6262626262626263 pt)=(1.0,0.6887253293507516,0.0),
        rgb(0.6363636363636364 pt)=(1.0,0.7093135311890413,0.0),
        rgb(0.6464646464646465 pt)=(1.0,0.7401958339464758,0.0),
        rgb(0.6565656565656566 pt)=(1.0,0.7710781367039102,0.0),
        rgb(0.6666666666666666 pt)=(1.0,0.7916663385421999,0.0),
        rgb(0.6767676767676768 pt)=(1.0,0.8225486412996345,0.0),
        rgb(0.6868686868686869 pt)=(1.0,0.843136843137924,0.0),
        rgb(0.696969696969697 pt)=(1.0,0.8740191458953586,0.0),
        rgb(0.7070707070707071 pt)=(1.0,0.9049014486527931,0.0),
        rgb(0.7171717171717171 pt)=(1.0,0.9254896504910827,0.0),
        rgb(0.7272727272727273 pt)=(1.0,0.9563719532485172,0.0),
        rgb(0.7373737373737373 pt)=(1.0,0.9769601550868069,0.0),
        rgb(0.7474747474747475 pt)=(1.0,1.0,0.011763717646070686),
        rgb(0.7575757575757576 pt)=(1.0,1.0,0.04264610146963098),
        rgb(0.7676767676767676 pt)=(1.0,1.0,0.08896967720497098),
        rgb(0.7777777777777778 pt)=(1.0,1.0,0.13529325294031186),
        rgb(0.7878787878787878 pt)=(1.0,1.0,0.16617563676387215),
        rgb(0.797979797979798 pt)=(1.0,1.0,0.21249921249921258),
        rgb(0.8080808080808081 pt)=(1.0,1.0,0.24338159632277287),
        rgb(0.8181818181818182 pt)=(1.0,1.0,0.2897051720581133),
        rgb(0.8282828282828283 pt)=(1.0,1.0,0.3360287477934533),
        rgb(0.8383838383838383 pt)=(1.0,1.0,0.36691113161701405),
        rgb(0.8484848484848485 pt)=(1.0,1.0,0.41323470735235446),
        rgb(0.8585858585858586 pt)=(1.0,1.0,0.44411709117591475),
        rgb(0.8686868686868687 pt)=(1.0,1.0,0.4904406669112552),
        rgb(0.8787878787878788 pt)=(1.0,1.0,0.5213230507348154),
        rgb(0.8888888888888888 pt)=(1.0,1.0,0.567646626470156),
        rgb(0.898989898989899 pt)=(1.0,1.0,0.6139702022054964),
        rgb(0.9090909090909091 pt)=(1.0,1.0,0.6448525860290567),
        rgb(0.9191919191919192 pt)=(1.0,1.0,0.6911761617643971),
        rgb(0.9292929292929293 pt)=(1.0,1.0,0.7220585455879573),
        rgb(0.9393939393939394 pt)=(1.0,1.0,0.7683821213232979),
        rgb(0.9494949494949495 pt)=(1.0,1.0,0.8147056970586383),
        rgb(0.9595959595959596 pt)=(1.0,1.0,0.8455880808821985),
        rgb(0.9696969696969697 pt)=(1.0,1.0,0.891911656617539),
        rgb(0.9797979797979798 pt)=(1.0,1.0,0.9227940404410993),
        rgb(0.98989898989899 pt)=(1.0,1.0,0.9691176161764397),
        rgb(1.0 pt)=(1.0,1.0,1.0),},
    label style = {font=\small},
    xticklabel style = {font=\tiny},
    colorbar style ={
        tick label style = {font=\tiny},
        xticklabel={$10^{
        \pgfmathparse{\tick}
        \pgfmathprintnumber[precision=2]{\pgfmathresult}}$},
        at = {(1.03,1)},
        anchor = north west,
    }
    ]
    \addplot graphics [
        xmin = -3.141593,
        xmax = 3.141593,
        ymin = -0.785398,
        ymax = 0.785398
    ] {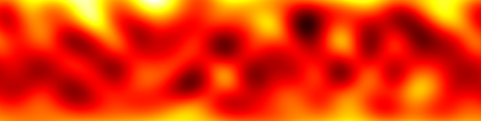};
\end{axis}
            \begin{axis}[
    title=,
    title style = {
        font=\bfseries,
        at = {(0.5, 0.9)},
        anchor = south
    },
    name=plot335,
    at=(plot334.below south),
    anchor=north,
    enlargelimits = false,
    axis on top = true,
    axis equal image,
    point meta min = -0.449031,
    point meta max = 2.098015,
    colorbar,xtick=\empty,ytick=\empty,
    colormap={hot}{
        rgb(0.0 pt)=(0.0416,0.0,0.0),
        rgb(0.010101010101010102 pt)=(0.0621896881088697,0.0,0.0),
        rgb(0.020202020202020204 pt)=(0.09307422027217424,0.0,0.0),
        rgb(0.030303030303030304 pt)=(0.11366390838104393,0.0,0.0),
        rgb(0.04040404040404041 pt)=(0.14454844054434846,0.0,0.0),
        rgb(0.050505050505050504 pt)=(0.1651381286532182,0.0,0.0),
        rgb(0.06060606060606061 pt)=(0.1960226608165227,0.0,0.0),
        rgb(0.0707070707070707 pt)=(0.22690719297982725,0.0,0.0),
        rgb(0.08080808080808081 pt)=(0.24749688108869694,0.0,0.0),
        rgb(0.09090909090909091 pt)=(0.2783814132520015,0.0,0.0),
        rgb(0.10101010101010101 pt)=(0.2989711013608711,0.0,0.0),
        rgb(0.1111111111111111 pt)=(0.3298556335241757,0.0,0.0),
        rgb(0.12121212121212122 pt)=(0.36074016568748024,0.0,0.0),
        rgb(0.13131313131313133 pt)=(0.38132985379634987,0.0,0.0),
        rgb(0.1414141414141414 pt)=(0.41221438595965454,0.0,0.0),
        rgb(0.15151515151515152 pt)=(0.43280407406852417,0.0,0.0),
        rgb(0.16161616161616163 pt)=(0.4636886062318286,0.0,0.0),
        rgb(0.1717171717171717 pt)=(0.48427829434069847,0.0,0.0),
        rgb(0.18181818181818182 pt)=(0.5151628265040029,0.0,0.0),
        rgb(0.1919191919191919 pt)=(0.5460473586673074,0.0,0.0),
        rgb(0.20202020202020202 pt)=(0.5666370467761772,0.0,0.0),
        rgb(0.21212121212121213 pt)=(0.5975215789394817,0.0,0.0),
        rgb(0.2222222222222222 pt)=(0.6181112670483514,0.0,0.0),
        rgb(0.23232323232323232 pt)=(0.6489957992116561,0.0,0.0),
        rgb(0.24242424242424243 pt)=(0.6798803313749605,0.0,0.0),
        rgb(0.25252525252525254 pt)=(0.7004700194838303,0.0,0.0),
        rgb(0.26262626262626265 pt)=(0.7313545516471347,0.0,0.0),
        rgb(0.2727272727272727 pt)=(0.7519442397560044,0.0,0.0),
        rgb(0.2828282828282828 pt)=(0.782828771919309,0.0,0.0),
        rgb(0.29292929292929293 pt)=(0.8034184600281785,0.0,0.0),
        rgb(0.30303030303030304 pt)=(0.8343029921914833,0.0,0.0),
        rgb(0.31313131313131315 pt)=(0.8651875243547877,0.0,0.0),
        rgb(0.32323232323232326 pt)=(0.8857772124636573,0.0,0.0),
        rgb(0.3333333333333333 pt)=(0.9166617446269619,0.0,0.0),
        rgb(0.3434343434343434 pt)=(0.9372514327358317,0.0,0.0),
        rgb(0.35353535353535354 pt)=(0.9681359648991361,0.0,0.0),
        rgb(0.36363636363636365 pt)=(0.9990204970624408,0.0,0.0),
        rgb(0.37373737373737376 pt)=(1.0,0.019608769606337603,0.0),
        rgb(0.3838383838383838 pt)=(1.0,0.05049107236377195,0.0),
        rgb(0.3939393939393939 pt)=(1.0,0.07107927420206175,0.0),
        rgb(0.40404040404040403 pt)=(1.0,0.10196157695949624,0.0),
        rgb(0.41414141414141414 pt)=(1.0,0.1328438797169306,0.0),
        rgb(0.42424242424242425 pt)=(1.0,0.1534320815552204,0.0),
        rgb(0.43434343434343436 pt)=(1.0,0.1843143843126549,0.0),
        rgb(0.4444444444444444 pt)=(1.0,0.20490258615094456,0.0),
        rgb(0.45454545454545453 pt)=(1.0,0.23578488890837906,0.0),
        rgb(0.46464646464646464 pt)=(1.0,0.2563730907466687,0.0),
        rgb(0.47474747474747475 pt)=(1.0,0.28725539350410323,0.0),
        rgb(0.48484848484848486 pt)=(1.0,0.3181376962615377,0.0),
        rgb(0.494949494949495 pt)=(1.0,0.33872589809982734,0.0),
        rgb(0.5050505050505051 pt)=(1.0,0.36960820085726187,0.0),
        rgb(0.5151515151515151 pt)=(1.0,0.3901964026955515,0.0),
        rgb(0.5252525252525253 pt)=(1.0,0.421078705452986,0.0),
        rgb(0.5353535353535354 pt)=(1.0,0.4519610082104205,0.0),
        rgb(0.5454545454545454 pt)=(1.0,0.47254921004871014,0.0),
        rgb(0.5555555555555556 pt)=(1.0,0.5034315128061446,0.0),
        rgb(0.5656565656565656 pt)=(1.0,0.5240197146444343,0.0),
        rgb(0.5757575757575758 pt)=(1.0,0.5549020174018688,0.0),
        rgb(0.5858585858585859 pt)=(1.0,0.5754902192401584,0.0),
        rgb(0.5959595959595959 pt)=(1.0,0.606372521997593,0.0),
        rgb(0.6060606060606061 pt)=(1.0,0.6372548247550275,0.0),
        rgb(0.6161616161616161 pt)=(1.0,0.6578430265933172,0.0),
        rgb(0.6262626262626263 pt)=(1.0,0.6887253293507516,0.0),
        rgb(0.6363636363636364 pt)=(1.0,0.7093135311890413,0.0),
        rgb(0.6464646464646465 pt)=(1.0,0.7401958339464758,0.0),
        rgb(0.6565656565656566 pt)=(1.0,0.7710781367039102,0.0),
        rgb(0.6666666666666666 pt)=(1.0,0.7916663385421999,0.0),
        rgb(0.6767676767676768 pt)=(1.0,0.8225486412996345,0.0),
        rgb(0.6868686868686869 pt)=(1.0,0.843136843137924,0.0),
        rgb(0.696969696969697 pt)=(1.0,0.8740191458953586,0.0),
        rgb(0.7070707070707071 pt)=(1.0,0.9049014486527931,0.0),
        rgb(0.7171717171717171 pt)=(1.0,0.9254896504910827,0.0),
        rgb(0.7272727272727273 pt)=(1.0,0.9563719532485172,0.0),
        rgb(0.7373737373737373 pt)=(1.0,0.9769601550868069,0.0),
        rgb(0.7474747474747475 pt)=(1.0,1.0,0.011763717646070686),
        rgb(0.7575757575757576 pt)=(1.0,1.0,0.04264610146963098),
        rgb(0.7676767676767676 pt)=(1.0,1.0,0.08896967720497098),
        rgb(0.7777777777777778 pt)=(1.0,1.0,0.13529325294031186),
        rgb(0.7878787878787878 pt)=(1.0,1.0,0.16617563676387215),
        rgb(0.797979797979798 pt)=(1.0,1.0,0.21249921249921258),
        rgb(0.8080808080808081 pt)=(1.0,1.0,0.24338159632277287),
        rgb(0.8181818181818182 pt)=(1.0,1.0,0.2897051720581133),
        rgb(0.8282828282828283 pt)=(1.0,1.0,0.3360287477934533),
        rgb(0.8383838383838383 pt)=(1.0,1.0,0.36691113161701405),
        rgb(0.8484848484848485 pt)=(1.0,1.0,0.41323470735235446),
        rgb(0.8585858585858586 pt)=(1.0,1.0,0.44411709117591475),
        rgb(0.8686868686868687 pt)=(1.0,1.0,0.4904406669112552),
        rgb(0.8787878787878788 pt)=(1.0,1.0,0.5213230507348154),
        rgb(0.8888888888888888 pt)=(1.0,1.0,0.567646626470156),
        rgb(0.898989898989899 pt)=(1.0,1.0,0.6139702022054964),
        rgb(0.9090909090909091 pt)=(1.0,1.0,0.6448525860290567),
        rgb(0.9191919191919192 pt)=(1.0,1.0,0.6911761617643971),
        rgb(0.9292929292929293 pt)=(1.0,1.0,0.7220585455879573),
        rgb(0.9393939393939394 pt)=(1.0,1.0,0.7683821213232979),
        rgb(0.9494949494949495 pt)=(1.0,1.0,0.8147056970586383),
        rgb(0.9595959595959596 pt)=(1.0,1.0,0.8455880808821985),
        rgb(0.9696969696969697 pt)=(1.0,1.0,0.891911656617539),
        rgb(0.9797979797979798 pt)=(1.0,1.0,0.9227940404410993),
        rgb(0.98989898989899 pt)=(1.0,1.0,0.9691176161764397),
        rgb(1.0 pt)=(1.0,1.0,1.0),},
    label style = {font=\small},
    xticklabel style = {font=\tiny},
    colorbar style ={
        tick label style = {font=\tiny},
        xticklabel={$10^{
        \pgfmathparse{\tick}
        \pgfmathprintnumber[precision=2]{\pgfmathresult}}$},
        at = {(1.03,1)},
        anchor = north west,
    }
    ]
    \addplot graphics [
        xmin = -3.141593,
        xmax = 3.141593,
        ymin = -0.785398,
        ymax = 0.785398
    ] {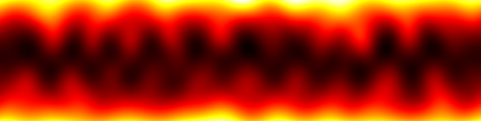};
\end{axis}
            \begin{axis}[
    title=,
    title style = {
        font=\bfseries,
        at = {(0.5, 0.9)},
        anchor = south
    },
    name=plot336,
    at=(plot335.below south),
    anchor=north,
    enlargelimits = false,
    axis on top = true,
    axis equal image,
    point meta min = -1.215294,
    point meta max = 1.407569,
    colorbar,xlabel = {az. [rad]},ytick=\empty,
    colormap={hot}{
        rgb(0.0 pt)=(0.0416,0.0,0.0),
        rgb(0.010101010101010102 pt)=(0.0621896881088697,0.0,0.0),
        rgb(0.020202020202020204 pt)=(0.09307422027217424,0.0,0.0),
        rgb(0.030303030303030304 pt)=(0.11366390838104393,0.0,0.0),
        rgb(0.04040404040404041 pt)=(0.14454844054434846,0.0,0.0),
        rgb(0.050505050505050504 pt)=(0.1651381286532182,0.0,0.0),
        rgb(0.06060606060606061 pt)=(0.1960226608165227,0.0,0.0),
        rgb(0.0707070707070707 pt)=(0.22690719297982725,0.0,0.0),
        rgb(0.08080808080808081 pt)=(0.24749688108869694,0.0,0.0),
        rgb(0.09090909090909091 pt)=(0.2783814132520015,0.0,0.0),
        rgb(0.10101010101010101 pt)=(0.2989711013608711,0.0,0.0),
        rgb(0.1111111111111111 pt)=(0.3298556335241757,0.0,0.0),
        rgb(0.12121212121212122 pt)=(0.36074016568748024,0.0,0.0),
        rgb(0.13131313131313133 pt)=(0.38132985379634987,0.0,0.0),
        rgb(0.1414141414141414 pt)=(0.41221438595965454,0.0,0.0),
        rgb(0.15151515151515152 pt)=(0.43280407406852417,0.0,0.0),
        rgb(0.16161616161616163 pt)=(0.4636886062318286,0.0,0.0),
        rgb(0.1717171717171717 pt)=(0.48427829434069847,0.0,0.0),
        rgb(0.18181818181818182 pt)=(0.5151628265040029,0.0,0.0),
        rgb(0.1919191919191919 pt)=(0.5460473586673074,0.0,0.0),
        rgb(0.20202020202020202 pt)=(0.5666370467761772,0.0,0.0),
        rgb(0.21212121212121213 pt)=(0.5975215789394817,0.0,0.0),
        rgb(0.2222222222222222 pt)=(0.6181112670483514,0.0,0.0),
        rgb(0.23232323232323232 pt)=(0.6489957992116561,0.0,0.0),
        rgb(0.24242424242424243 pt)=(0.6798803313749605,0.0,0.0),
        rgb(0.25252525252525254 pt)=(0.7004700194838303,0.0,0.0),
        rgb(0.26262626262626265 pt)=(0.7313545516471347,0.0,0.0),
        rgb(0.2727272727272727 pt)=(0.7519442397560044,0.0,0.0),
        rgb(0.2828282828282828 pt)=(0.782828771919309,0.0,0.0),
        rgb(0.29292929292929293 pt)=(0.8034184600281785,0.0,0.0),
        rgb(0.30303030303030304 pt)=(0.8343029921914833,0.0,0.0),
        rgb(0.31313131313131315 pt)=(0.8651875243547877,0.0,0.0),
        rgb(0.32323232323232326 pt)=(0.8857772124636573,0.0,0.0),
        rgb(0.3333333333333333 pt)=(0.9166617446269619,0.0,0.0),
        rgb(0.3434343434343434 pt)=(0.9372514327358317,0.0,0.0),
        rgb(0.35353535353535354 pt)=(0.9681359648991361,0.0,0.0),
        rgb(0.36363636363636365 pt)=(0.9990204970624408,0.0,0.0),
        rgb(0.37373737373737376 pt)=(1.0,0.019608769606337603,0.0),
        rgb(0.3838383838383838 pt)=(1.0,0.05049107236377195,0.0),
        rgb(0.3939393939393939 pt)=(1.0,0.07107927420206175,0.0),
        rgb(0.40404040404040403 pt)=(1.0,0.10196157695949624,0.0),
        rgb(0.41414141414141414 pt)=(1.0,0.1328438797169306,0.0),
        rgb(0.42424242424242425 pt)=(1.0,0.1534320815552204,0.0),
        rgb(0.43434343434343436 pt)=(1.0,0.1843143843126549,0.0),
        rgb(0.4444444444444444 pt)=(1.0,0.20490258615094456,0.0),
        rgb(0.45454545454545453 pt)=(1.0,0.23578488890837906,0.0),
        rgb(0.46464646464646464 pt)=(1.0,0.2563730907466687,0.0),
        rgb(0.47474747474747475 pt)=(1.0,0.28725539350410323,0.0),
        rgb(0.48484848484848486 pt)=(1.0,0.3181376962615377,0.0),
        rgb(0.494949494949495 pt)=(1.0,0.33872589809982734,0.0),
        rgb(0.5050505050505051 pt)=(1.0,0.36960820085726187,0.0),
        rgb(0.5151515151515151 pt)=(1.0,0.3901964026955515,0.0),
        rgb(0.5252525252525253 pt)=(1.0,0.421078705452986,0.0),
        rgb(0.5353535353535354 pt)=(1.0,0.4519610082104205,0.0),
        rgb(0.5454545454545454 pt)=(1.0,0.47254921004871014,0.0),
        rgb(0.5555555555555556 pt)=(1.0,0.5034315128061446,0.0),
        rgb(0.5656565656565656 pt)=(1.0,0.5240197146444343,0.0),
        rgb(0.5757575757575758 pt)=(1.0,0.5549020174018688,0.0),
        rgb(0.5858585858585859 pt)=(1.0,0.5754902192401584,0.0),
        rgb(0.5959595959595959 pt)=(1.0,0.606372521997593,0.0),
        rgb(0.6060606060606061 pt)=(1.0,0.6372548247550275,0.0),
        rgb(0.6161616161616161 pt)=(1.0,0.6578430265933172,0.0),
        rgb(0.6262626262626263 pt)=(1.0,0.6887253293507516,0.0),
        rgb(0.6363636363636364 pt)=(1.0,0.7093135311890413,0.0),
        rgb(0.6464646464646465 pt)=(1.0,0.7401958339464758,0.0),
        rgb(0.6565656565656566 pt)=(1.0,0.7710781367039102,0.0),
        rgb(0.6666666666666666 pt)=(1.0,0.7916663385421999,0.0),
        rgb(0.6767676767676768 pt)=(1.0,0.8225486412996345,0.0),
        rgb(0.6868686868686869 pt)=(1.0,0.843136843137924,0.0),
        rgb(0.696969696969697 pt)=(1.0,0.8740191458953586,0.0),
        rgb(0.7070707070707071 pt)=(1.0,0.9049014486527931,0.0),
        rgb(0.7171717171717171 pt)=(1.0,0.9254896504910827,0.0),
        rgb(0.7272727272727273 pt)=(1.0,0.9563719532485172,0.0),
        rgb(0.7373737373737373 pt)=(1.0,0.9769601550868069,0.0),
        rgb(0.7474747474747475 pt)=(1.0,1.0,0.011763717646070686),
        rgb(0.7575757575757576 pt)=(1.0,1.0,0.04264610146963098),
        rgb(0.7676767676767676 pt)=(1.0,1.0,0.08896967720497098),
        rgb(0.7777777777777778 pt)=(1.0,1.0,0.13529325294031186),
        rgb(0.7878787878787878 pt)=(1.0,1.0,0.16617563676387215),
        rgb(0.797979797979798 pt)=(1.0,1.0,0.21249921249921258),
        rgb(0.8080808080808081 pt)=(1.0,1.0,0.24338159632277287),
        rgb(0.8181818181818182 pt)=(1.0,1.0,0.2897051720581133),
        rgb(0.8282828282828283 pt)=(1.0,1.0,0.3360287477934533),
        rgb(0.8383838383838383 pt)=(1.0,1.0,0.36691113161701405),
        rgb(0.8484848484848485 pt)=(1.0,1.0,0.41323470735235446),
        rgb(0.8585858585858586 pt)=(1.0,1.0,0.44411709117591475),
        rgb(0.8686868686868687 pt)=(1.0,1.0,0.4904406669112552),
        rgb(0.8787878787878788 pt)=(1.0,1.0,0.5213230507348154),
        rgb(0.8888888888888888 pt)=(1.0,1.0,0.567646626470156),
        rgb(0.898989898989899 pt)=(1.0,1.0,0.6139702022054964),
        rgb(0.9090909090909091 pt)=(1.0,1.0,0.6448525860290567),
        rgb(0.9191919191919192 pt)=(1.0,1.0,0.6911761617643971),
        rgb(0.9292929292929293 pt)=(1.0,1.0,0.7220585455879573),
        rgb(0.9393939393939394 pt)=(1.0,1.0,0.7683821213232979),
        rgb(0.9494949494949495 pt)=(1.0,1.0,0.8147056970586383),
        rgb(0.9595959595959596 pt)=(1.0,1.0,0.8455880808821985),
        rgb(0.9696969696969697 pt)=(1.0,1.0,0.891911656617539),
        rgb(0.9797979797979798 pt)=(1.0,1.0,0.9227940404410993),
        rgb(0.98989898989899 pt)=(1.0,1.0,0.9691176161764397),
        rgb(1.0 pt)=(1.0,1.0,1.0),},
    label style = {font=\small},
    xticklabel style = {font=\tiny},
    colorbar style ={
        tick label style = {font=\tiny},
        xticklabel={$10^{
        \pgfmathparse{\tick}
        \pgfmathprintnumber[precision=2]{\pgfmathresult}}$},
        at = {(1.03,1)},
        anchor = north west,
    }
    ]
    \addplot graphics [
        xmin = -3.141593,
        xmax = 3.141593,
        ymin = -0.785398,
        ymax = 0.785398
    ] {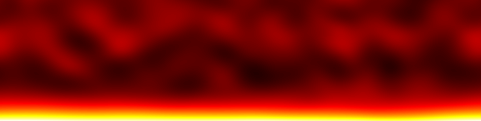};
\end{axis}
            \begin{axis}[
    title=SGD-based Approach,
    title style = {
        font=\bfseries,
        at = {(0.5, 0.9)},
        anchor = south
    },
    name=plot337,
    at=(plot333.below south),
    anchor=above north,
    enlargelimits = false,
    axis on top = true,
    axis equal image,
    point meta min = -2.360772,
    point meta max = -1.318539,
    ,xtick=\empty,ylabel = {ele. [rad]},
    colormap={hot}{
        rgb(0.0 pt)=(0.0416,0.0,0.0),
        rgb(0.010101010101010102 pt)=(0.0621896881088697,0.0,0.0),
        rgb(0.020202020202020204 pt)=(0.09307422027217424,0.0,0.0),
        rgb(0.030303030303030304 pt)=(0.11366390838104393,0.0,0.0),
        rgb(0.04040404040404041 pt)=(0.14454844054434846,0.0,0.0),
        rgb(0.050505050505050504 pt)=(0.1651381286532182,0.0,0.0),
        rgb(0.06060606060606061 pt)=(0.1960226608165227,0.0,0.0),
        rgb(0.0707070707070707 pt)=(0.22690719297982725,0.0,0.0),
        rgb(0.08080808080808081 pt)=(0.24749688108869694,0.0,0.0),
        rgb(0.09090909090909091 pt)=(0.2783814132520015,0.0,0.0),
        rgb(0.10101010101010101 pt)=(0.2989711013608711,0.0,0.0),
        rgb(0.1111111111111111 pt)=(0.3298556335241757,0.0,0.0),
        rgb(0.12121212121212122 pt)=(0.36074016568748024,0.0,0.0),
        rgb(0.13131313131313133 pt)=(0.38132985379634987,0.0,0.0),
        rgb(0.1414141414141414 pt)=(0.41221438595965454,0.0,0.0),
        rgb(0.15151515151515152 pt)=(0.43280407406852417,0.0,0.0),
        rgb(0.16161616161616163 pt)=(0.4636886062318286,0.0,0.0),
        rgb(0.1717171717171717 pt)=(0.48427829434069847,0.0,0.0),
        rgb(0.18181818181818182 pt)=(0.5151628265040029,0.0,0.0),
        rgb(0.1919191919191919 pt)=(0.5460473586673074,0.0,0.0),
        rgb(0.20202020202020202 pt)=(0.5666370467761772,0.0,0.0),
        rgb(0.21212121212121213 pt)=(0.5975215789394817,0.0,0.0),
        rgb(0.2222222222222222 pt)=(0.6181112670483514,0.0,0.0),
        rgb(0.23232323232323232 pt)=(0.6489957992116561,0.0,0.0),
        rgb(0.24242424242424243 pt)=(0.6798803313749605,0.0,0.0),
        rgb(0.25252525252525254 pt)=(0.7004700194838303,0.0,0.0),
        rgb(0.26262626262626265 pt)=(0.7313545516471347,0.0,0.0),
        rgb(0.2727272727272727 pt)=(0.7519442397560044,0.0,0.0),
        rgb(0.2828282828282828 pt)=(0.782828771919309,0.0,0.0),
        rgb(0.29292929292929293 pt)=(0.8034184600281785,0.0,0.0),
        rgb(0.30303030303030304 pt)=(0.8343029921914833,0.0,0.0),
        rgb(0.31313131313131315 pt)=(0.8651875243547877,0.0,0.0),
        rgb(0.32323232323232326 pt)=(0.8857772124636573,0.0,0.0),
        rgb(0.3333333333333333 pt)=(0.9166617446269619,0.0,0.0),
        rgb(0.3434343434343434 pt)=(0.9372514327358317,0.0,0.0),
        rgb(0.35353535353535354 pt)=(0.9681359648991361,0.0,0.0),
        rgb(0.36363636363636365 pt)=(0.9990204970624408,0.0,0.0),
        rgb(0.37373737373737376 pt)=(1.0,0.019608769606337603,0.0),
        rgb(0.3838383838383838 pt)=(1.0,0.05049107236377195,0.0),
        rgb(0.3939393939393939 pt)=(1.0,0.07107927420206175,0.0),
        rgb(0.40404040404040403 pt)=(1.0,0.10196157695949624,0.0),
        rgb(0.41414141414141414 pt)=(1.0,0.1328438797169306,0.0),
        rgb(0.42424242424242425 pt)=(1.0,0.1534320815552204,0.0),
        rgb(0.43434343434343436 pt)=(1.0,0.1843143843126549,0.0),
        rgb(0.4444444444444444 pt)=(1.0,0.20490258615094456,0.0),
        rgb(0.45454545454545453 pt)=(1.0,0.23578488890837906,0.0),
        rgb(0.46464646464646464 pt)=(1.0,0.2563730907466687,0.0),
        rgb(0.47474747474747475 pt)=(1.0,0.28725539350410323,0.0),
        rgb(0.48484848484848486 pt)=(1.0,0.3181376962615377,0.0),
        rgb(0.494949494949495 pt)=(1.0,0.33872589809982734,0.0),
        rgb(0.5050505050505051 pt)=(1.0,0.36960820085726187,0.0),
        rgb(0.5151515151515151 pt)=(1.0,0.3901964026955515,0.0),
        rgb(0.5252525252525253 pt)=(1.0,0.421078705452986,0.0),
        rgb(0.5353535353535354 pt)=(1.0,0.4519610082104205,0.0),
        rgb(0.5454545454545454 pt)=(1.0,0.47254921004871014,0.0),
        rgb(0.5555555555555556 pt)=(1.0,0.5034315128061446,0.0),
        rgb(0.5656565656565656 pt)=(1.0,0.5240197146444343,0.0),
        rgb(0.5757575757575758 pt)=(1.0,0.5549020174018688,0.0),
        rgb(0.5858585858585859 pt)=(1.0,0.5754902192401584,0.0),
        rgb(0.5959595959595959 pt)=(1.0,0.606372521997593,0.0),
        rgb(0.6060606060606061 pt)=(1.0,0.6372548247550275,0.0),
        rgb(0.6161616161616161 pt)=(1.0,0.6578430265933172,0.0),
        rgb(0.6262626262626263 pt)=(1.0,0.6887253293507516,0.0),
        rgb(0.6363636363636364 pt)=(1.0,0.7093135311890413,0.0),
        rgb(0.6464646464646465 pt)=(1.0,0.7401958339464758,0.0),
        rgb(0.6565656565656566 pt)=(1.0,0.7710781367039102,0.0),
        rgb(0.6666666666666666 pt)=(1.0,0.7916663385421999,0.0),
        rgb(0.6767676767676768 pt)=(1.0,0.8225486412996345,0.0),
        rgb(0.6868686868686869 pt)=(1.0,0.843136843137924,0.0),
        rgb(0.696969696969697 pt)=(1.0,0.8740191458953586,0.0),
        rgb(0.7070707070707071 pt)=(1.0,0.9049014486527931,0.0),
        rgb(0.7171717171717171 pt)=(1.0,0.9254896504910827,0.0),
        rgb(0.7272727272727273 pt)=(1.0,0.9563719532485172,0.0),
        rgb(0.7373737373737373 pt)=(1.0,0.9769601550868069,0.0),
        rgb(0.7474747474747475 pt)=(1.0,1.0,0.011763717646070686),
        rgb(0.7575757575757576 pt)=(1.0,1.0,0.04264610146963098),
        rgb(0.7676767676767676 pt)=(1.0,1.0,0.08896967720497098),
        rgb(0.7777777777777778 pt)=(1.0,1.0,0.13529325294031186),
        rgb(0.7878787878787878 pt)=(1.0,1.0,0.16617563676387215),
        rgb(0.797979797979798 pt)=(1.0,1.0,0.21249921249921258),
        rgb(0.8080808080808081 pt)=(1.0,1.0,0.24338159632277287),
        rgb(0.8181818181818182 pt)=(1.0,1.0,0.2897051720581133),
        rgb(0.8282828282828283 pt)=(1.0,1.0,0.3360287477934533),
        rgb(0.8383838383838383 pt)=(1.0,1.0,0.36691113161701405),
        rgb(0.8484848484848485 pt)=(1.0,1.0,0.41323470735235446),
        rgb(0.8585858585858586 pt)=(1.0,1.0,0.44411709117591475),
        rgb(0.8686868686868687 pt)=(1.0,1.0,0.4904406669112552),
        rgb(0.8787878787878788 pt)=(1.0,1.0,0.5213230507348154),
        rgb(0.8888888888888888 pt)=(1.0,1.0,0.567646626470156),
        rgb(0.898989898989899 pt)=(1.0,1.0,0.6139702022054964),
        rgb(0.9090909090909091 pt)=(1.0,1.0,0.6448525860290567),
        rgb(0.9191919191919192 pt)=(1.0,1.0,0.6911761617643971),
        rgb(0.9292929292929293 pt)=(1.0,1.0,0.7220585455879573),
        rgb(0.9393939393939394 pt)=(1.0,1.0,0.7683821213232979),
        rgb(0.9494949494949495 pt)=(1.0,1.0,0.8147056970586383),
        rgb(0.9595959595959596 pt)=(1.0,1.0,0.8455880808821985),
        rgb(0.9696969696969697 pt)=(1.0,1.0,0.891911656617539),
        rgb(0.9797979797979798 pt)=(1.0,1.0,0.9227940404410993),
        rgb(0.98989898989899 pt)=(1.0,1.0,0.9691176161764397),
        rgb(1.0 pt)=(1.0,1.0,1.0),},
    label style = {font=\small},
    xticklabel style = {font=\tiny},
    colorbar style ={
        tick label style = {font=\tiny},
        xticklabel={$10^{
        \pgfmathparse{\tick}
        \pgfmathprintnumber[precision=2]{\pgfmathresult}}$},
        at = {(1.03,1)},
        anchor = north west,
    }
    ]
    \addplot graphics [
        xmin = -3.141593,
        xmax = 3.141593,
        ymin = -0.785398,
        ymax = 0.785398
    ] {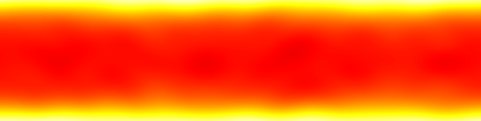};
\end{axis}
            \begin{axis}[
    title=,
    title style = {
        font=\bfseries,
        at = {(0.5, 0.9)},
        anchor = south
    },
    name=plot338,
    at=(plot337.below south),
    anchor=north,
    enlargelimits = false,
    axis on top = true,
    axis equal image,
    point meta min = -0.449031,
    point meta max = 2.196963,
    ,xtick=\empty,ylabel = {ele. [rad]},
    colormap={hot}{
        rgb(0.0 pt)=(0.0416,0.0,0.0),
        rgb(0.010101010101010102 pt)=(0.0621896881088697,0.0,0.0),
        rgb(0.020202020202020204 pt)=(0.09307422027217424,0.0,0.0),
        rgb(0.030303030303030304 pt)=(0.11366390838104393,0.0,0.0),
        rgb(0.04040404040404041 pt)=(0.14454844054434846,0.0,0.0),
        rgb(0.050505050505050504 pt)=(0.1651381286532182,0.0,0.0),
        rgb(0.06060606060606061 pt)=(0.1960226608165227,0.0,0.0),
        rgb(0.0707070707070707 pt)=(0.22690719297982725,0.0,0.0),
        rgb(0.08080808080808081 pt)=(0.24749688108869694,0.0,0.0),
        rgb(0.09090909090909091 pt)=(0.2783814132520015,0.0,0.0),
        rgb(0.10101010101010101 pt)=(0.2989711013608711,0.0,0.0),
        rgb(0.1111111111111111 pt)=(0.3298556335241757,0.0,0.0),
        rgb(0.12121212121212122 pt)=(0.36074016568748024,0.0,0.0),
        rgb(0.13131313131313133 pt)=(0.38132985379634987,0.0,0.0),
        rgb(0.1414141414141414 pt)=(0.41221438595965454,0.0,0.0),
        rgb(0.15151515151515152 pt)=(0.43280407406852417,0.0,0.0),
        rgb(0.16161616161616163 pt)=(0.4636886062318286,0.0,0.0),
        rgb(0.1717171717171717 pt)=(0.48427829434069847,0.0,0.0),
        rgb(0.18181818181818182 pt)=(0.5151628265040029,0.0,0.0),
        rgb(0.1919191919191919 pt)=(0.5460473586673074,0.0,0.0),
        rgb(0.20202020202020202 pt)=(0.5666370467761772,0.0,0.0),
        rgb(0.21212121212121213 pt)=(0.5975215789394817,0.0,0.0),
        rgb(0.2222222222222222 pt)=(0.6181112670483514,0.0,0.0),
        rgb(0.23232323232323232 pt)=(0.6489957992116561,0.0,0.0),
        rgb(0.24242424242424243 pt)=(0.6798803313749605,0.0,0.0),
        rgb(0.25252525252525254 pt)=(0.7004700194838303,0.0,0.0),
        rgb(0.26262626262626265 pt)=(0.7313545516471347,0.0,0.0),
        rgb(0.2727272727272727 pt)=(0.7519442397560044,0.0,0.0),
        rgb(0.2828282828282828 pt)=(0.782828771919309,0.0,0.0),
        rgb(0.29292929292929293 pt)=(0.8034184600281785,0.0,0.0),
        rgb(0.30303030303030304 pt)=(0.8343029921914833,0.0,0.0),
        rgb(0.31313131313131315 pt)=(0.8651875243547877,0.0,0.0),
        rgb(0.32323232323232326 pt)=(0.8857772124636573,0.0,0.0),
        rgb(0.3333333333333333 pt)=(0.9166617446269619,0.0,0.0),
        rgb(0.3434343434343434 pt)=(0.9372514327358317,0.0,0.0),
        rgb(0.35353535353535354 pt)=(0.9681359648991361,0.0,0.0),
        rgb(0.36363636363636365 pt)=(0.9990204970624408,0.0,0.0),
        rgb(0.37373737373737376 pt)=(1.0,0.019608769606337603,0.0),
        rgb(0.3838383838383838 pt)=(1.0,0.05049107236377195,0.0),
        rgb(0.3939393939393939 pt)=(1.0,0.07107927420206175,0.0),
        rgb(0.40404040404040403 pt)=(1.0,0.10196157695949624,0.0),
        rgb(0.41414141414141414 pt)=(1.0,0.1328438797169306,0.0),
        rgb(0.42424242424242425 pt)=(1.0,0.1534320815552204,0.0),
        rgb(0.43434343434343436 pt)=(1.0,0.1843143843126549,0.0),
        rgb(0.4444444444444444 pt)=(1.0,0.20490258615094456,0.0),
        rgb(0.45454545454545453 pt)=(1.0,0.23578488890837906,0.0),
        rgb(0.46464646464646464 pt)=(1.0,0.2563730907466687,0.0),
        rgb(0.47474747474747475 pt)=(1.0,0.28725539350410323,0.0),
        rgb(0.48484848484848486 pt)=(1.0,0.3181376962615377,0.0),
        rgb(0.494949494949495 pt)=(1.0,0.33872589809982734,0.0),
        rgb(0.5050505050505051 pt)=(1.0,0.36960820085726187,0.0),
        rgb(0.5151515151515151 pt)=(1.0,0.3901964026955515,0.0),
        rgb(0.5252525252525253 pt)=(1.0,0.421078705452986,0.0),
        rgb(0.5353535353535354 pt)=(1.0,0.4519610082104205,0.0),
        rgb(0.5454545454545454 pt)=(1.0,0.47254921004871014,0.0),
        rgb(0.5555555555555556 pt)=(1.0,0.5034315128061446,0.0),
        rgb(0.5656565656565656 pt)=(1.0,0.5240197146444343,0.0),
        rgb(0.5757575757575758 pt)=(1.0,0.5549020174018688,0.0),
        rgb(0.5858585858585859 pt)=(1.0,0.5754902192401584,0.0),
        rgb(0.5959595959595959 pt)=(1.0,0.606372521997593,0.0),
        rgb(0.6060606060606061 pt)=(1.0,0.6372548247550275,0.0),
        rgb(0.6161616161616161 pt)=(1.0,0.6578430265933172,0.0),
        rgb(0.6262626262626263 pt)=(1.0,0.6887253293507516,0.0),
        rgb(0.6363636363636364 pt)=(1.0,0.7093135311890413,0.0),
        rgb(0.6464646464646465 pt)=(1.0,0.7401958339464758,0.0),
        rgb(0.6565656565656566 pt)=(1.0,0.7710781367039102,0.0),
        rgb(0.6666666666666666 pt)=(1.0,0.7916663385421999,0.0),
        rgb(0.6767676767676768 pt)=(1.0,0.8225486412996345,0.0),
        rgb(0.6868686868686869 pt)=(1.0,0.843136843137924,0.0),
        rgb(0.696969696969697 pt)=(1.0,0.8740191458953586,0.0),
        rgb(0.7070707070707071 pt)=(1.0,0.9049014486527931,0.0),
        rgb(0.7171717171717171 pt)=(1.0,0.9254896504910827,0.0),
        rgb(0.7272727272727273 pt)=(1.0,0.9563719532485172,0.0),
        rgb(0.7373737373737373 pt)=(1.0,0.9769601550868069,0.0),
        rgb(0.7474747474747475 pt)=(1.0,1.0,0.011763717646070686),
        rgb(0.7575757575757576 pt)=(1.0,1.0,0.04264610146963098),
        rgb(0.7676767676767676 pt)=(1.0,1.0,0.08896967720497098),
        rgb(0.7777777777777778 pt)=(1.0,1.0,0.13529325294031186),
        rgb(0.7878787878787878 pt)=(1.0,1.0,0.16617563676387215),
        rgb(0.797979797979798 pt)=(1.0,1.0,0.21249921249921258),
        rgb(0.8080808080808081 pt)=(1.0,1.0,0.24338159632277287),
        rgb(0.8181818181818182 pt)=(1.0,1.0,0.2897051720581133),
        rgb(0.8282828282828283 pt)=(1.0,1.0,0.3360287477934533),
        rgb(0.8383838383838383 pt)=(1.0,1.0,0.36691113161701405),
        rgb(0.8484848484848485 pt)=(1.0,1.0,0.41323470735235446),
        rgb(0.8585858585858586 pt)=(1.0,1.0,0.44411709117591475),
        rgb(0.8686868686868687 pt)=(1.0,1.0,0.4904406669112552),
        rgb(0.8787878787878788 pt)=(1.0,1.0,0.5213230507348154),
        rgb(0.8888888888888888 pt)=(1.0,1.0,0.567646626470156),
        rgb(0.898989898989899 pt)=(1.0,1.0,0.6139702022054964),
        rgb(0.9090909090909091 pt)=(1.0,1.0,0.6448525860290567),
        rgb(0.9191919191919192 pt)=(1.0,1.0,0.6911761617643971),
        rgb(0.9292929292929293 pt)=(1.0,1.0,0.7220585455879573),
        rgb(0.9393939393939394 pt)=(1.0,1.0,0.7683821213232979),
        rgb(0.9494949494949495 pt)=(1.0,1.0,0.8147056970586383),
        rgb(0.9595959595959596 pt)=(1.0,1.0,0.8455880808821985),
        rgb(0.9696969696969697 pt)=(1.0,1.0,0.891911656617539),
        rgb(0.9797979797979798 pt)=(1.0,1.0,0.9227940404410993),
        rgb(0.98989898989899 pt)=(1.0,1.0,0.9691176161764397),
        rgb(1.0 pt)=(1.0,1.0,1.0),},
    label style = {font=\small},
    xticklabel style = {font=\tiny},
    colorbar style ={
        tick label style = {font=\tiny},
        xticklabel={$10^{
        \pgfmathparse{\tick}
        \pgfmathprintnumber[precision=2]{\pgfmathresult}}$},
        at = {(1.03,1)},
        anchor = north west,
    }
    ]
    \addplot graphics [
        xmin = -3.141593,
        xmax = 3.141593,
        ymin = -0.785398,
        ymax = 0.785398
    ] {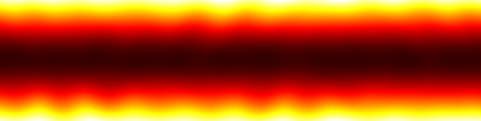};
\end{axis}
            \begin{axis}[
    title=,
    title style = {
        font=\bfseries,
        at = {(0.5, 0.9)},
        anchor = south
    },
    name=plot339,
    at=(plot338.below south),
    anchor=north,
    enlargelimits = false,
    axis on top = true,
    axis equal image,
    point meta min = -1.215294,
    point meta max = 2.928957,
    ,xlabel = {az. [rad]},ylabel = {ele. [rad]},
    colormap={hot}{
        rgb(0.0 pt)=(0.0416,0.0,0.0),
        rgb(0.010101010101010102 pt)=(0.0621896881088697,0.0,0.0),
        rgb(0.020202020202020204 pt)=(0.09307422027217424,0.0,0.0),
        rgb(0.030303030303030304 pt)=(0.11366390838104393,0.0,0.0),
        rgb(0.04040404040404041 pt)=(0.14454844054434846,0.0,0.0),
        rgb(0.050505050505050504 pt)=(0.1651381286532182,0.0,0.0),
        rgb(0.06060606060606061 pt)=(0.1960226608165227,0.0,0.0),
        rgb(0.0707070707070707 pt)=(0.22690719297982725,0.0,0.0),
        rgb(0.08080808080808081 pt)=(0.24749688108869694,0.0,0.0),
        rgb(0.09090909090909091 pt)=(0.2783814132520015,0.0,0.0),
        rgb(0.10101010101010101 pt)=(0.2989711013608711,0.0,0.0),
        rgb(0.1111111111111111 pt)=(0.3298556335241757,0.0,0.0),
        rgb(0.12121212121212122 pt)=(0.36074016568748024,0.0,0.0),
        rgb(0.13131313131313133 pt)=(0.38132985379634987,0.0,0.0),
        rgb(0.1414141414141414 pt)=(0.41221438595965454,0.0,0.0),
        rgb(0.15151515151515152 pt)=(0.43280407406852417,0.0,0.0),
        rgb(0.16161616161616163 pt)=(0.4636886062318286,0.0,0.0),
        rgb(0.1717171717171717 pt)=(0.48427829434069847,0.0,0.0),
        rgb(0.18181818181818182 pt)=(0.5151628265040029,0.0,0.0),
        rgb(0.1919191919191919 pt)=(0.5460473586673074,0.0,0.0),
        rgb(0.20202020202020202 pt)=(0.5666370467761772,0.0,0.0),
        rgb(0.21212121212121213 pt)=(0.5975215789394817,0.0,0.0),
        rgb(0.2222222222222222 pt)=(0.6181112670483514,0.0,0.0),
        rgb(0.23232323232323232 pt)=(0.6489957992116561,0.0,0.0),
        rgb(0.24242424242424243 pt)=(0.6798803313749605,0.0,0.0),
        rgb(0.25252525252525254 pt)=(0.7004700194838303,0.0,0.0),
        rgb(0.26262626262626265 pt)=(0.7313545516471347,0.0,0.0),
        rgb(0.2727272727272727 pt)=(0.7519442397560044,0.0,0.0),
        rgb(0.2828282828282828 pt)=(0.782828771919309,0.0,0.0),
        rgb(0.29292929292929293 pt)=(0.8034184600281785,0.0,0.0),
        rgb(0.30303030303030304 pt)=(0.8343029921914833,0.0,0.0),
        rgb(0.31313131313131315 pt)=(0.8651875243547877,0.0,0.0),
        rgb(0.32323232323232326 pt)=(0.8857772124636573,0.0,0.0),
        rgb(0.3333333333333333 pt)=(0.9166617446269619,0.0,0.0),
        rgb(0.3434343434343434 pt)=(0.9372514327358317,0.0,0.0),
        rgb(0.35353535353535354 pt)=(0.9681359648991361,0.0,0.0),
        rgb(0.36363636363636365 pt)=(0.9990204970624408,0.0,0.0),
        rgb(0.37373737373737376 pt)=(1.0,0.019608769606337603,0.0),
        rgb(0.3838383838383838 pt)=(1.0,0.05049107236377195,0.0),
        rgb(0.3939393939393939 pt)=(1.0,0.07107927420206175,0.0),
        rgb(0.40404040404040403 pt)=(1.0,0.10196157695949624,0.0),
        rgb(0.41414141414141414 pt)=(1.0,0.1328438797169306,0.0),
        rgb(0.42424242424242425 pt)=(1.0,0.1534320815552204,0.0),
        rgb(0.43434343434343436 pt)=(1.0,0.1843143843126549,0.0),
        rgb(0.4444444444444444 pt)=(1.0,0.20490258615094456,0.0),
        rgb(0.45454545454545453 pt)=(1.0,0.23578488890837906,0.0),
        rgb(0.46464646464646464 pt)=(1.0,0.2563730907466687,0.0),
        rgb(0.47474747474747475 pt)=(1.0,0.28725539350410323,0.0),
        rgb(0.48484848484848486 pt)=(1.0,0.3181376962615377,0.0),
        rgb(0.494949494949495 pt)=(1.0,0.33872589809982734,0.0),
        rgb(0.5050505050505051 pt)=(1.0,0.36960820085726187,0.0),
        rgb(0.5151515151515151 pt)=(1.0,0.3901964026955515,0.0),
        rgb(0.5252525252525253 pt)=(1.0,0.421078705452986,0.0),
        rgb(0.5353535353535354 pt)=(1.0,0.4519610082104205,0.0),
        rgb(0.5454545454545454 pt)=(1.0,0.47254921004871014,0.0),
        rgb(0.5555555555555556 pt)=(1.0,0.5034315128061446,0.0),
        rgb(0.5656565656565656 pt)=(1.0,0.5240197146444343,0.0),
        rgb(0.5757575757575758 pt)=(1.0,0.5549020174018688,0.0),
        rgb(0.5858585858585859 pt)=(1.0,0.5754902192401584,0.0),
        rgb(0.5959595959595959 pt)=(1.0,0.606372521997593,0.0),
        rgb(0.6060606060606061 pt)=(1.0,0.6372548247550275,0.0),
        rgb(0.6161616161616161 pt)=(1.0,0.6578430265933172,0.0),
        rgb(0.6262626262626263 pt)=(1.0,0.6887253293507516,0.0),
        rgb(0.6363636363636364 pt)=(1.0,0.7093135311890413,0.0),
        rgb(0.6464646464646465 pt)=(1.0,0.7401958339464758,0.0),
        rgb(0.6565656565656566 pt)=(1.0,0.7710781367039102,0.0),
        rgb(0.6666666666666666 pt)=(1.0,0.7916663385421999,0.0),
        rgb(0.6767676767676768 pt)=(1.0,0.8225486412996345,0.0),
        rgb(0.6868686868686869 pt)=(1.0,0.843136843137924,0.0),
        rgb(0.696969696969697 pt)=(1.0,0.8740191458953586,0.0),
        rgb(0.7070707070707071 pt)=(1.0,0.9049014486527931,0.0),
        rgb(0.7171717171717171 pt)=(1.0,0.9254896504910827,0.0),
        rgb(0.7272727272727273 pt)=(1.0,0.9563719532485172,0.0),
        rgb(0.7373737373737373 pt)=(1.0,0.9769601550868069,0.0),
        rgb(0.7474747474747475 pt)=(1.0,1.0,0.011763717646070686),
        rgb(0.7575757575757576 pt)=(1.0,1.0,0.04264610146963098),
        rgb(0.7676767676767676 pt)=(1.0,1.0,0.08896967720497098),
        rgb(0.7777777777777778 pt)=(1.0,1.0,0.13529325294031186),
        rgb(0.7878787878787878 pt)=(1.0,1.0,0.16617563676387215),
        rgb(0.797979797979798 pt)=(1.0,1.0,0.21249921249921258),
        rgb(0.8080808080808081 pt)=(1.0,1.0,0.24338159632277287),
        rgb(0.8181818181818182 pt)=(1.0,1.0,0.2897051720581133),
        rgb(0.8282828282828283 pt)=(1.0,1.0,0.3360287477934533),
        rgb(0.8383838383838383 pt)=(1.0,1.0,0.36691113161701405),
        rgb(0.8484848484848485 pt)=(1.0,1.0,0.41323470735235446),
        rgb(0.8585858585858586 pt)=(1.0,1.0,0.44411709117591475),
        rgb(0.8686868686868687 pt)=(1.0,1.0,0.4904406669112552),
        rgb(0.8787878787878788 pt)=(1.0,1.0,0.5213230507348154),
        rgb(0.8888888888888888 pt)=(1.0,1.0,0.567646626470156),
        rgb(0.898989898989899 pt)=(1.0,1.0,0.6139702022054964),
        rgb(0.9090909090909091 pt)=(1.0,1.0,0.6448525860290567),
        rgb(0.9191919191919192 pt)=(1.0,1.0,0.6911761617643971),
        rgb(0.9292929292929293 pt)=(1.0,1.0,0.7220585455879573),
        rgb(0.9393939393939394 pt)=(1.0,1.0,0.7683821213232979),
        rgb(0.9494949494949495 pt)=(1.0,1.0,0.8147056970586383),
        rgb(0.9595959595959596 pt)=(1.0,1.0,0.8455880808821985),
        rgb(0.9696969696969697 pt)=(1.0,1.0,0.891911656617539),
        rgb(0.9797979797979798 pt)=(1.0,1.0,0.9227940404410993),
        rgb(0.98989898989899 pt)=(1.0,1.0,0.9691176161764397),
        rgb(1.0 pt)=(1.0,1.0,1.0),},
    label style = {font=\small},
    xticklabel style = {font=\tiny},
    colorbar style ={
        tick label style = {font=\tiny},
        xticklabel={$10^{
        \pgfmathparse{\tick}
        \pgfmathprintnumber[precision=2]{\pgfmathresult}}$},
        at = {(1.03,1)},
        anchor = north west,
    }
    ]
    \addplot graphics [
        xmin = -3.141593,
        xmax = 3.141593,
        ymin = -0.785398,
        ymax = 0.785398
    ] {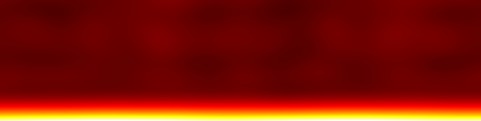};
\end{axis}
            \begin{axis}[
    title=Uncompressed,
    title style = {
        font=\bfseries,
        at = {(0.5, 0.9)},
        anchor = south
    },
    name=plot3310,
    at=(plot336.below south),
    anchor=above north,
    enlargelimits = false,
    axis on top = true,
    axis equal image,
    point meta min = -2.360772,
    point meta max = -1.318539,
    colorbar,xtick=\empty,ytick=\empty,
    colormap={hot}{
        rgb(0.0 pt)=(0.0416,0.0,0.0),
        rgb(0.010101010101010102 pt)=(0.0621896881088697,0.0,0.0),
        rgb(0.020202020202020204 pt)=(0.09307422027217424,0.0,0.0),
        rgb(0.030303030303030304 pt)=(0.11366390838104393,0.0,0.0),
        rgb(0.04040404040404041 pt)=(0.14454844054434846,0.0,0.0),
        rgb(0.050505050505050504 pt)=(0.1651381286532182,0.0,0.0),
        rgb(0.06060606060606061 pt)=(0.1960226608165227,0.0,0.0),
        rgb(0.0707070707070707 pt)=(0.22690719297982725,0.0,0.0),
        rgb(0.08080808080808081 pt)=(0.24749688108869694,0.0,0.0),
        rgb(0.09090909090909091 pt)=(0.2783814132520015,0.0,0.0),
        rgb(0.10101010101010101 pt)=(0.2989711013608711,0.0,0.0),
        rgb(0.1111111111111111 pt)=(0.3298556335241757,0.0,0.0),
        rgb(0.12121212121212122 pt)=(0.36074016568748024,0.0,0.0),
        rgb(0.13131313131313133 pt)=(0.38132985379634987,0.0,0.0),
        rgb(0.1414141414141414 pt)=(0.41221438595965454,0.0,0.0),
        rgb(0.15151515151515152 pt)=(0.43280407406852417,0.0,0.0),
        rgb(0.16161616161616163 pt)=(0.4636886062318286,0.0,0.0),
        rgb(0.1717171717171717 pt)=(0.48427829434069847,0.0,0.0),
        rgb(0.18181818181818182 pt)=(0.5151628265040029,0.0,0.0),
        rgb(0.1919191919191919 pt)=(0.5460473586673074,0.0,0.0),
        rgb(0.20202020202020202 pt)=(0.5666370467761772,0.0,0.0),
        rgb(0.21212121212121213 pt)=(0.5975215789394817,0.0,0.0),
        rgb(0.2222222222222222 pt)=(0.6181112670483514,0.0,0.0),
        rgb(0.23232323232323232 pt)=(0.6489957992116561,0.0,0.0),
        rgb(0.24242424242424243 pt)=(0.6798803313749605,0.0,0.0),
        rgb(0.25252525252525254 pt)=(0.7004700194838303,0.0,0.0),
        rgb(0.26262626262626265 pt)=(0.7313545516471347,0.0,0.0),
        rgb(0.2727272727272727 pt)=(0.7519442397560044,0.0,0.0),
        rgb(0.2828282828282828 pt)=(0.782828771919309,0.0,0.0),
        rgb(0.29292929292929293 pt)=(0.8034184600281785,0.0,0.0),
        rgb(0.30303030303030304 pt)=(0.8343029921914833,0.0,0.0),
        rgb(0.31313131313131315 pt)=(0.8651875243547877,0.0,0.0),
        rgb(0.32323232323232326 pt)=(0.8857772124636573,0.0,0.0),
        rgb(0.3333333333333333 pt)=(0.9166617446269619,0.0,0.0),
        rgb(0.3434343434343434 pt)=(0.9372514327358317,0.0,0.0),
        rgb(0.35353535353535354 pt)=(0.9681359648991361,0.0,0.0),
        rgb(0.36363636363636365 pt)=(0.9990204970624408,0.0,0.0),
        rgb(0.37373737373737376 pt)=(1.0,0.019608769606337603,0.0),
        rgb(0.3838383838383838 pt)=(1.0,0.05049107236377195,0.0),
        rgb(0.3939393939393939 pt)=(1.0,0.07107927420206175,0.0),
        rgb(0.40404040404040403 pt)=(1.0,0.10196157695949624,0.0),
        rgb(0.41414141414141414 pt)=(1.0,0.1328438797169306,0.0),
        rgb(0.42424242424242425 pt)=(1.0,0.1534320815552204,0.0),
        rgb(0.43434343434343436 pt)=(1.0,0.1843143843126549,0.0),
        rgb(0.4444444444444444 pt)=(1.0,0.20490258615094456,0.0),
        rgb(0.45454545454545453 pt)=(1.0,0.23578488890837906,0.0),
        rgb(0.46464646464646464 pt)=(1.0,0.2563730907466687,0.0),
        rgb(0.47474747474747475 pt)=(1.0,0.28725539350410323,0.0),
        rgb(0.48484848484848486 pt)=(1.0,0.3181376962615377,0.0),
        rgb(0.494949494949495 pt)=(1.0,0.33872589809982734,0.0),
        rgb(0.5050505050505051 pt)=(1.0,0.36960820085726187,0.0),
        rgb(0.5151515151515151 pt)=(1.0,0.3901964026955515,0.0),
        rgb(0.5252525252525253 pt)=(1.0,0.421078705452986,0.0),
        rgb(0.5353535353535354 pt)=(1.0,0.4519610082104205,0.0),
        rgb(0.5454545454545454 pt)=(1.0,0.47254921004871014,0.0),
        rgb(0.5555555555555556 pt)=(1.0,0.5034315128061446,0.0),
        rgb(0.5656565656565656 pt)=(1.0,0.5240197146444343,0.0),
        rgb(0.5757575757575758 pt)=(1.0,0.5549020174018688,0.0),
        rgb(0.5858585858585859 pt)=(1.0,0.5754902192401584,0.0),
        rgb(0.5959595959595959 pt)=(1.0,0.606372521997593,0.0),
        rgb(0.6060606060606061 pt)=(1.0,0.6372548247550275,0.0),
        rgb(0.6161616161616161 pt)=(1.0,0.6578430265933172,0.0),
        rgb(0.6262626262626263 pt)=(1.0,0.6887253293507516,0.0),
        rgb(0.6363636363636364 pt)=(1.0,0.7093135311890413,0.0),
        rgb(0.6464646464646465 pt)=(1.0,0.7401958339464758,0.0),
        rgb(0.6565656565656566 pt)=(1.0,0.7710781367039102,0.0),
        rgb(0.6666666666666666 pt)=(1.0,0.7916663385421999,0.0),
        rgb(0.6767676767676768 pt)=(1.0,0.8225486412996345,0.0),
        rgb(0.6868686868686869 pt)=(1.0,0.843136843137924,0.0),
        rgb(0.696969696969697 pt)=(1.0,0.8740191458953586,0.0),
        rgb(0.7070707070707071 pt)=(1.0,0.9049014486527931,0.0),
        rgb(0.7171717171717171 pt)=(1.0,0.9254896504910827,0.0),
        rgb(0.7272727272727273 pt)=(1.0,0.9563719532485172,0.0),
        rgb(0.7373737373737373 pt)=(1.0,0.9769601550868069,0.0),
        rgb(0.7474747474747475 pt)=(1.0,1.0,0.011763717646070686),
        rgb(0.7575757575757576 pt)=(1.0,1.0,0.04264610146963098),
        rgb(0.7676767676767676 pt)=(1.0,1.0,0.08896967720497098),
        rgb(0.7777777777777778 pt)=(1.0,1.0,0.13529325294031186),
        rgb(0.7878787878787878 pt)=(1.0,1.0,0.16617563676387215),
        rgb(0.797979797979798 pt)=(1.0,1.0,0.21249921249921258),
        rgb(0.8080808080808081 pt)=(1.0,1.0,0.24338159632277287),
        rgb(0.8181818181818182 pt)=(1.0,1.0,0.2897051720581133),
        rgb(0.8282828282828283 pt)=(1.0,1.0,0.3360287477934533),
        rgb(0.8383838383838383 pt)=(1.0,1.0,0.36691113161701405),
        rgb(0.8484848484848485 pt)=(1.0,1.0,0.41323470735235446),
        rgb(0.8585858585858586 pt)=(1.0,1.0,0.44411709117591475),
        rgb(0.8686868686868687 pt)=(1.0,1.0,0.4904406669112552),
        rgb(0.8787878787878788 pt)=(1.0,1.0,0.5213230507348154),
        rgb(0.8888888888888888 pt)=(1.0,1.0,0.567646626470156),
        rgb(0.898989898989899 pt)=(1.0,1.0,0.6139702022054964),
        rgb(0.9090909090909091 pt)=(1.0,1.0,0.6448525860290567),
        rgb(0.9191919191919192 pt)=(1.0,1.0,0.6911761617643971),
        rgb(0.9292929292929293 pt)=(1.0,1.0,0.7220585455879573),
        rgb(0.9393939393939394 pt)=(1.0,1.0,0.7683821213232979),
        rgb(0.9494949494949495 pt)=(1.0,1.0,0.8147056970586383),
        rgb(0.9595959595959596 pt)=(1.0,1.0,0.8455880808821985),
        rgb(0.9696969696969697 pt)=(1.0,1.0,0.891911656617539),
        rgb(0.9797979797979798 pt)=(1.0,1.0,0.9227940404410993),
        rgb(0.98989898989899 pt)=(1.0,1.0,0.9691176161764397),
        rgb(1.0 pt)=(1.0,1.0,1.0),},
    label style = {font=\small},
    xticklabel style = {font=\tiny},
    colorbar style ={
        tick label style = {font=\tiny},
        xticklabel={$10^{
        \pgfmathparse{\tick}
        \pgfmathprintnumber[precision=2]{\pgfmathresult}}$},
        at = {(1.03,1)},
        anchor = north west,
    }
    ]
    \addplot graphics [
        xmin = -3.141593,
        xmax = 3.141593,
        ymin = -0.785398,
        ymax = 0.785398
    ] {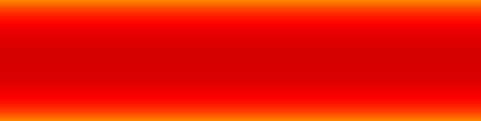};
\end{axis}
            \begin{axis}[
    title=,
    title style = {
        font=\bfseries,
        at = {(0.5, 0.9)},
        anchor = south
    },
    name=plot3311,
    at=(plot3310.below south),
    anchor=north,
    enlargelimits = false,
    axis on top = true,
    axis equal image,
    point meta min = -0.449031,
    point meta max = 2.196963,
    colorbar,xtick=\empty,ytick=\empty,
    colormap={hot}{
        rgb(0.0 pt)=(0.0416,0.0,0.0),
        rgb(0.010101010101010102 pt)=(0.0621896881088697,0.0,0.0),
        rgb(0.020202020202020204 pt)=(0.09307422027217424,0.0,0.0),
        rgb(0.030303030303030304 pt)=(0.11366390838104393,0.0,0.0),
        rgb(0.04040404040404041 pt)=(0.14454844054434846,0.0,0.0),
        rgb(0.050505050505050504 pt)=(0.1651381286532182,0.0,0.0),
        rgb(0.06060606060606061 pt)=(0.1960226608165227,0.0,0.0),
        rgb(0.0707070707070707 pt)=(0.22690719297982725,0.0,0.0),
        rgb(0.08080808080808081 pt)=(0.24749688108869694,0.0,0.0),
        rgb(0.09090909090909091 pt)=(0.2783814132520015,0.0,0.0),
        rgb(0.10101010101010101 pt)=(0.2989711013608711,0.0,0.0),
        rgb(0.1111111111111111 pt)=(0.3298556335241757,0.0,0.0),
        rgb(0.12121212121212122 pt)=(0.36074016568748024,0.0,0.0),
        rgb(0.13131313131313133 pt)=(0.38132985379634987,0.0,0.0),
        rgb(0.1414141414141414 pt)=(0.41221438595965454,0.0,0.0),
        rgb(0.15151515151515152 pt)=(0.43280407406852417,0.0,0.0),
        rgb(0.16161616161616163 pt)=(0.4636886062318286,0.0,0.0),
        rgb(0.1717171717171717 pt)=(0.48427829434069847,0.0,0.0),
        rgb(0.18181818181818182 pt)=(0.5151628265040029,0.0,0.0),
        rgb(0.1919191919191919 pt)=(0.5460473586673074,0.0,0.0),
        rgb(0.20202020202020202 pt)=(0.5666370467761772,0.0,0.0),
        rgb(0.21212121212121213 pt)=(0.5975215789394817,0.0,0.0),
        rgb(0.2222222222222222 pt)=(0.6181112670483514,0.0,0.0),
        rgb(0.23232323232323232 pt)=(0.6489957992116561,0.0,0.0),
        rgb(0.24242424242424243 pt)=(0.6798803313749605,0.0,0.0),
        rgb(0.25252525252525254 pt)=(0.7004700194838303,0.0,0.0),
        rgb(0.26262626262626265 pt)=(0.7313545516471347,0.0,0.0),
        rgb(0.2727272727272727 pt)=(0.7519442397560044,0.0,0.0),
        rgb(0.2828282828282828 pt)=(0.782828771919309,0.0,0.0),
        rgb(0.29292929292929293 pt)=(0.8034184600281785,0.0,0.0),
        rgb(0.30303030303030304 pt)=(0.8343029921914833,0.0,0.0),
        rgb(0.31313131313131315 pt)=(0.8651875243547877,0.0,0.0),
        rgb(0.32323232323232326 pt)=(0.8857772124636573,0.0,0.0),
        rgb(0.3333333333333333 pt)=(0.9166617446269619,0.0,0.0),
        rgb(0.3434343434343434 pt)=(0.9372514327358317,0.0,0.0),
        rgb(0.35353535353535354 pt)=(0.9681359648991361,0.0,0.0),
        rgb(0.36363636363636365 pt)=(0.9990204970624408,0.0,0.0),
        rgb(0.37373737373737376 pt)=(1.0,0.019608769606337603,0.0),
        rgb(0.3838383838383838 pt)=(1.0,0.05049107236377195,0.0),
        rgb(0.3939393939393939 pt)=(1.0,0.07107927420206175,0.0),
        rgb(0.40404040404040403 pt)=(1.0,0.10196157695949624,0.0),
        rgb(0.41414141414141414 pt)=(1.0,0.1328438797169306,0.0),
        rgb(0.42424242424242425 pt)=(1.0,0.1534320815552204,0.0),
        rgb(0.43434343434343436 pt)=(1.0,0.1843143843126549,0.0),
        rgb(0.4444444444444444 pt)=(1.0,0.20490258615094456,0.0),
        rgb(0.45454545454545453 pt)=(1.0,0.23578488890837906,0.0),
        rgb(0.46464646464646464 pt)=(1.0,0.2563730907466687,0.0),
        rgb(0.47474747474747475 pt)=(1.0,0.28725539350410323,0.0),
        rgb(0.48484848484848486 pt)=(1.0,0.3181376962615377,0.0),
        rgb(0.494949494949495 pt)=(1.0,0.33872589809982734,0.0),
        rgb(0.5050505050505051 pt)=(1.0,0.36960820085726187,0.0),
        rgb(0.5151515151515151 pt)=(1.0,0.3901964026955515,0.0),
        rgb(0.5252525252525253 pt)=(1.0,0.421078705452986,0.0),
        rgb(0.5353535353535354 pt)=(1.0,0.4519610082104205,0.0),
        rgb(0.5454545454545454 pt)=(1.0,0.47254921004871014,0.0),
        rgb(0.5555555555555556 pt)=(1.0,0.5034315128061446,0.0),
        rgb(0.5656565656565656 pt)=(1.0,0.5240197146444343,0.0),
        rgb(0.5757575757575758 pt)=(1.0,0.5549020174018688,0.0),
        rgb(0.5858585858585859 pt)=(1.0,0.5754902192401584,0.0),
        rgb(0.5959595959595959 pt)=(1.0,0.606372521997593,0.0),
        rgb(0.6060606060606061 pt)=(1.0,0.6372548247550275,0.0),
        rgb(0.6161616161616161 pt)=(1.0,0.6578430265933172,0.0),
        rgb(0.6262626262626263 pt)=(1.0,0.6887253293507516,0.0),
        rgb(0.6363636363636364 pt)=(1.0,0.7093135311890413,0.0),
        rgb(0.6464646464646465 pt)=(1.0,0.7401958339464758,0.0),
        rgb(0.6565656565656566 pt)=(1.0,0.7710781367039102,0.0),
        rgb(0.6666666666666666 pt)=(1.0,0.7916663385421999,0.0),
        rgb(0.6767676767676768 pt)=(1.0,0.8225486412996345,0.0),
        rgb(0.6868686868686869 pt)=(1.0,0.843136843137924,0.0),
        rgb(0.696969696969697 pt)=(1.0,0.8740191458953586,0.0),
        rgb(0.7070707070707071 pt)=(1.0,0.9049014486527931,0.0),
        rgb(0.7171717171717171 pt)=(1.0,0.9254896504910827,0.0),
        rgb(0.7272727272727273 pt)=(1.0,0.9563719532485172,0.0),
        rgb(0.7373737373737373 pt)=(1.0,0.9769601550868069,0.0),
        rgb(0.7474747474747475 pt)=(1.0,1.0,0.011763717646070686),
        rgb(0.7575757575757576 pt)=(1.0,1.0,0.04264610146963098),
        rgb(0.7676767676767676 pt)=(1.0,1.0,0.08896967720497098),
        rgb(0.7777777777777778 pt)=(1.0,1.0,0.13529325294031186),
        rgb(0.7878787878787878 pt)=(1.0,1.0,0.16617563676387215),
        rgb(0.797979797979798 pt)=(1.0,1.0,0.21249921249921258),
        rgb(0.8080808080808081 pt)=(1.0,1.0,0.24338159632277287),
        rgb(0.8181818181818182 pt)=(1.0,1.0,0.2897051720581133),
        rgb(0.8282828282828283 pt)=(1.0,1.0,0.3360287477934533),
        rgb(0.8383838383838383 pt)=(1.0,1.0,0.36691113161701405),
        rgb(0.8484848484848485 pt)=(1.0,1.0,0.41323470735235446),
        rgb(0.8585858585858586 pt)=(1.0,1.0,0.44411709117591475),
        rgb(0.8686868686868687 pt)=(1.0,1.0,0.4904406669112552),
        rgb(0.8787878787878788 pt)=(1.0,1.0,0.5213230507348154),
        rgb(0.8888888888888888 pt)=(1.0,1.0,0.567646626470156),
        rgb(0.898989898989899 pt)=(1.0,1.0,0.6139702022054964),
        rgb(0.9090909090909091 pt)=(1.0,1.0,0.6448525860290567),
        rgb(0.9191919191919192 pt)=(1.0,1.0,0.6911761617643971),
        rgb(0.9292929292929293 pt)=(1.0,1.0,0.7220585455879573),
        rgb(0.9393939393939394 pt)=(1.0,1.0,0.7683821213232979),
        rgb(0.9494949494949495 pt)=(1.0,1.0,0.8147056970586383),
        rgb(0.9595959595959596 pt)=(1.0,1.0,0.8455880808821985),
        rgb(0.9696969696969697 pt)=(1.0,1.0,0.891911656617539),
        rgb(0.9797979797979798 pt)=(1.0,1.0,0.9227940404410993),
        rgb(0.98989898989899 pt)=(1.0,1.0,0.9691176161764397),
        rgb(1.0 pt)=(1.0,1.0,1.0),},
    label style = {font=\small},
    xticklabel style = {font=\tiny},
    colorbar style ={
        tick label style = {font=\tiny},
        xticklabel={$10^{
        \pgfmathparse{\tick}
        \pgfmathprintnumber[precision=2]{\pgfmathresult}}$},
        at = {(1.03,1)},
        anchor = north west,
    }
    ]
    \addplot graphics [
        xmin = -3.141593,
        xmax = 3.141593,
        ymin = -0.785398,
        ymax = 0.785398
    ] {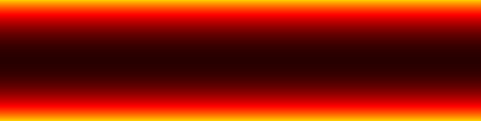};
\end{axis}
            \begin{axis}[
    title=,
    title style = {
        font=\bfseries,
        at = {(0.5, 0.9)},
        anchor = south
    },
    name=plot3312,
    at=(plot3311.below south),
    anchor=north,
    enlargelimits = false,
    axis on top = true,
    axis equal image,
    point meta min = -1.215294,
    point meta max = 2.928957,
    colorbar,xlabel = {az. [rad]},ytick=\empty,
    colormap={hot}{
        rgb(0.0 pt)=(0.0416,0.0,0.0),
        rgb(0.010101010101010102 pt)=(0.0621896881088697,0.0,0.0),
        rgb(0.020202020202020204 pt)=(0.09307422027217424,0.0,0.0),
        rgb(0.030303030303030304 pt)=(0.11366390838104393,0.0,0.0),
        rgb(0.04040404040404041 pt)=(0.14454844054434846,0.0,0.0),
        rgb(0.050505050505050504 pt)=(0.1651381286532182,0.0,0.0),
        rgb(0.06060606060606061 pt)=(0.1960226608165227,0.0,0.0),
        rgb(0.0707070707070707 pt)=(0.22690719297982725,0.0,0.0),
        rgb(0.08080808080808081 pt)=(0.24749688108869694,0.0,0.0),
        rgb(0.09090909090909091 pt)=(0.2783814132520015,0.0,0.0),
        rgb(0.10101010101010101 pt)=(0.2989711013608711,0.0,0.0),
        rgb(0.1111111111111111 pt)=(0.3298556335241757,0.0,0.0),
        rgb(0.12121212121212122 pt)=(0.36074016568748024,0.0,0.0),
        rgb(0.13131313131313133 pt)=(0.38132985379634987,0.0,0.0),
        rgb(0.1414141414141414 pt)=(0.41221438595965454,0.0,0.0),
        rgb(0.15151515151515152 pt)=(0.43280407406852417,0.0,0.0),
        rgb(0.16161616161616163 pt)=(0.4636886062318286,0.0,0.0),
        rgb(0.1717171717171717 pt)=(0.48427829434069847,0.0,0.0),
        rgb(0.18181818181818182 pt)=(0.5151628265040029,0.0,0.0),
        rgb(0.1919191919191919 pt)=(0.5460473586673074,0.0,0.0),
        rgb(0.20202020202020202 pt)=(0.5666370467761772,0.0,0.0),
        rgb(0.21212121212121213 pt)=(0.5975215789394817,0.0,0.0),
        rgb(0.2222222222222222 pt)=(0.6181112670483514,0.0,0.0),
        rgb(0.23232323232323232 pt)=(0.6489957992116561,0.0,0.0),
        rgb(0.24242424242424243 pt)=(0.6798803313749605,0.0,0.0),
        rgb(0.25252525252525254 pt)=(0.7004700194838303,0.0,0.0),
        rgb(0.26262626262626265 pt)=(0.7313545516471347,0.0,0.0),
        rgb(0.2727272727272727 pt)=(0.7519442397560044,0.0,0.0),
        rgb(0.2828282828282828 pt)=(0.782828771919309,0.0,0.0),
        rgb(0.29292929292929293 pt)=(0.8034184600281785,0.0,0.0),
        rgb(0.30303030303030304 pt)=(0.8343029921914833,0.0,0.0),
        rgb(0.31313131313131315 pt)=(0.8651875243547877,0.0,0.0),
        rgb(0.32323232323232326 pt)=(0.8857772124636573,0.0,0.0),
        rgb(0.3333333333333333 pt)=(0.9166617446269619,0.0,0.0),
        rgb(0.3434343434343434 pt)=(0.9372514327358317,0.0,0.0),
        rgb(0.35353535353535354 pt)=(0.9681359648991361,0.0,0.0),
        rgb(0.36363636363636365 pt)=(0.9990204970624408,0.0,0.0),
        rgb(0.37373737373737376 pt)=(1.0,0.019608769606337603,0.0),
        rgb(0.3838383838383838 pt)=(1.0,0.05049107236377195,0.0),
        rgb(0.3939393939393939 pt)=(1.0,0.07107927420206175,0.0),
        rgb(0.40404040404040403 pt)=(1.0,0.10196157695949624,0.0),
        rgb(0.41414141414141414 pt)=(1.0,0.1328438797169306,0.0),
        rgb(0.42424242424242425 pt)=(1.0,0.1534320815552204,0.0),
        rgb(0.43434343434343436 pt)=(1.0,0.1843143843126549,0.0),
        rgb(0.4444444444444444 pt)=(1.0,0.20490258615094456,0.0),
        rgb(0.45454545454545453 pt)=(1.0,0.23578488890837906,0.0),
        rgb(0.46464646464646464 pt)=(1.0,0.2563730907466687,0.0),
        rgb(0.47474747474747475 pt)=(1.0,0.28725539350410323,0.0),
        rgb(0.48484848484848486 pt)=(1.0,0.3181376962615377,0.0),
        rgb(0.494949494949495 pt)=(1.0,0.33872589809982734,0.0),
        rgb(0.5050505050505051 pt)=(1.0,0.36960820085726187,0.0),
        rgb(0.5151515151515151 pt)=(1.0,0.3901964026955515,0.0),
        rgb(0.5252525252525253 pt)=(1.0,0.421078705452986,0.0),
        rgb(0.5353535353535354 pt)=(1.0,0.4519610082104205,0.0),
        rgb(0.5454545454545454 pt)=(1.0,0.47254921004871014,0.0),
        rgb(0.5555555555555556 pt)=(1.0,0.5034315128061446,0.0),
        rgb(0.5656565656565656 pt)=(1.0,0.5240197146444343,0.0),
        rgb(0.5757575757575758 pt)=(1.0,0.5549020174018688,0.0),
        rgb(0.5858585858585859 pt)=(1.0,0.5754902192401584,0.0),
        rgb(0.5959595959595959 pt)=(1.0,0.606372521997593,0.0),
        rgb(0.6060606060606061 pt)=(1.0,0.6372548247550275,0.0),
        rgb(0.6161616161616161 pt)=(1.0,0.6578430265933172,0.0),
        rgb(0.6262626262626263 pt)=(1.0,0.6887253293507516,0.0),
        rgb(0.6363636363636364 pt)=(1.0,0.7093135311890413,0.0),
        rgb(0.6464646464646465 pt)=(1.0,0.7401958339464758,0.0),
        rgb(0.6565656565656566 pt)=(1.0,0.7710781367039102,0.0),
        rgb(0.6666666666666666 pt)=(1.0,0.7916663385421999,0.0),
        rgb(0.6767676767676768 pt)=(1.0,0.8225486412996345,0.0),
        rgb(0.6868686868686869 pt)=(1.0,0.843136843137924,0.0),
        rgb(0.696969696969697 pt)=(1.0,0.8740191458953586,0.0),
        rgb(0.7070707070707071 pt)=(1.0,0.9049014486527931,0.0),
        rgb(0.7171717171717171 pt)=(1.0,0.9254896504910827,0.0),
        rgb(0.7272727272727273 pt)=(1.0,0.9563719532485172,0.0),
        rgb(0.7373737373737373 pt)=(1.0,0.9769601550868069,0.0),
        rgb(0.7474747474747475 pt)=(1.0,1.0,0.011763717646070686),
        rgb(0.7575757575757576 pt)=(1.0,1.0,0.04264610146963098),
        rgb(0.7676767676767676 pt)=(1.0,1.0,0.08896967720497098),
        rgb(0.7777777777777778 pt)=(1.0,1.0,0.13529325294031186),
        rgb(0.7878787878787878 pt)=(1.0,1.0,0.16617563676387215),
        rgb(0.797979797979798 pt)=(1.0,1.0,0.21249921249921258),
        rgb(0.8080808080808081 pt)=(1.0,1.0,0.24338159632277287),
        rgb(0.8181818181818182 pt)=(1.0,1.0,0.2897051720581133),
        rgb(0.8282828282828283 pt)=(1.0,1.0,0.3360287477934533),
        rgb(0.8383838383838383 pt)=(1.0,1.0,0.36691113161701405),
        rgb(0.8484848484848485 pt)=(1.0,1.0,0.41323470735235446),
        rgb(0.8585858585858586 pt)=(1.0,1.0,0.44411709117591475),
        rgb(0.8686868686868687 pt)=(1.0,1.0,0.4904406669112552),
        rgb(0.8787878787878788 pt)=(1.0,1.0,0.5213230507348154),
        rgb(0.8888888888888888 pt)=(1.0,1.0,0.567646626470156),
        rgb(0.898989898989899 pt)=(1.0,1.0,0.6139702022054964),
        rgb(0.9090909090909091 pt)=(1.0,1.0,0.6448525860290567),
        rgb(0.9191919191919192 pt)=(1.0,1.0,0.6911761617643971),
        rgb(0.9292929292929293 pt)=(1.0,1.0,0.7220585455879573),
        rgb(0.9393939393939394 pt)=(1.0,1.0,0.7683821213232979),
        rgb(0.9494949494949495 pt)=(1.0,1.0,0.8147056970586383),
        rgb(0.9595959595959596 pt)=(1.0,1.0,0.8455880808821985),
        rgb(0.9696969696969697 pt)=(1.0,1.0,0.891911656617539),
        rgb(0.9797979797979798 pt)=(1.0,1.0,0.9227940404410993),
        rgb(0.98989898989899 pt)=(1.0,1.0,0.9691176161764397),
        rgb(1.0 pt)=(1.0,1.0,1.0),},
    label style = {font=\small},
    xticklabel style = {font=\tiny},
    colorbar style ={
        tick label style = {font=\tiny},
        xticklabel={$10^{
        \pgfmathparse{\tick}
        \pgfmathprintnumber[precision=2]{\pgfmathresult}}$},
        at = {(1.03,1)},
        anchor = north west,
    }
    ]
    \addplot graphics [
        xmin = -3.141593,
        xmax = 3.141593,
        ymin = -0.785398,
        ymax = 0.785398
    ] {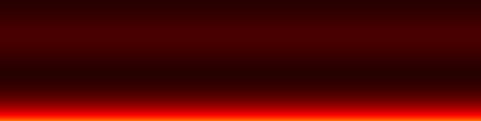};
\end{axis}
        \end{tikzpicture}
        \caption{\small Magnitude of the deterministic \gls{crb} on a logarithmic scale for fixed noise level $\sigma^2 = 1$ dependent on azimuth and elevation. \emph{Top}: single source, \emph{Middle}: two sources separated in azimuth by $2 \pi / 10$, \emph{Bottom}: two sources separated in elevation by $2 \pi / 10$. Notice the differing color bars for top, middle and bottom.}\label{crb_comparison}
    \end{center}
    \vspace{-8mm}
\end{figure*}

To quantify the performance of the proposed design apporach for the case of 2D \gls{doa} estimation, we also evaluate the deterministic \gls{crb} in Figure \ref{crb_comparison}, since it serves as a proxy to assess the possible performance of any unbiased estimator. For instance, the maximum likelihood estimator always reaches this lower bound asymptotically in the effective SNR, so one can expect that any well designed estimation procedure behaves similarly in the asymptotic regime. With spatial compression, the deterministic \gls{crb} for the $2$-dimensional case with $S$ sources and $1$ snapshot can be computed via~\cite{IRLKRLGDT:17}
\begin{align}
 C(\bm \theta) = \frac{\sigma^2}{2} \Tr\left(\left[ \Re(\bm D^\herm \bm \Pi_{\bm{G}}^\perp \bm D \odot (\bm{\mbox{1}}_{2 \times 2} \otimes \hat{\bm R})^\trans) \right]^{-1} \right), \label{eqn:crb_uncomp}
\end{align}
with $\bm \Pi_{\bm{G}}^\perp = \bm I - \bm G (\bm G^\herm \bm G)^{-1}\bm G^\herm$, $\otimes$ denoting the Kronecker product, $\Re$ the real part of a complex number and $\hat{\bm R} = \bm x \cdot \bm x^\herm$ being the sample covariance and we have set
\begin{align*}
\bm G &= \bm\Phi [\bm a_{\mathrm{SUCA}}(\theta_1, \vartheta_1), \dots, \bm a_{\mathrm{SUCA}}(\theta_S, \vartheta_S)],\\
\bm D_i &= \frac{\partial}{\partial \bm \theta_i} \bm G,~ \bm D = \left[\bm D_1, \dots, \bm D_d \right].
\end{align*}

The results in Figure \ref{crb_comparison} show for a fixed noise level $\sigma^2 = 1$ how the \gls{crb} changes for the four different sensing matrix designs depending on the position of a single source (top), two sources separated in azimuth (middle) and elevation (bottom). In the two sources case, the first source is located at the position denoted in the plot and the second with angular distance $2\pi / 10$ in azimuth or elevation. As one can see the random combining matrix and the previous approach introduce a highly varying sensitivity of the \gls{crb} with respect to azimuth and elevation, rendering the resulting combining matrices hard to apply for \gls{doa} estimation because of this non-uniformity in the angular domain. The \gls{sgd}-based approach results in a significantly smoother behavior of the \gls{crb} with a uniform increase across the whole azimuth and elevation region. Thus, the compressed array resulting from \gls{sgd} mimics the behavior of the uncompressed array more closely in terms of the \gls{crb}, which ultimately was the goal of the proposed design process.

\section{Conclusion and Outlook}
We have presented a flexible and computationally efficient scheme to optimize combining matrices for 2D \gls{doa} with compressive antenna arrays. First, it is not bound to a discrete dictionary for the angles, since they are selected during the gradient descent and do not have to reside on any grid. Second, it is possible to extend the model and the optimization to incorporate polarization or bistatic TX-RX setups with applications in channel sounding. Moreover, since the method is not bound to special antenna geometries, it can even cope with measured antenna patterns. Additionally, one has some level of control over the compressed array's behavior by choosing a suitable distribution for $\bm \Theta$. The low complexity also allows to use the proposed method as an online optimization during the measurement process itself by adaptively focusing on certain regions of the parameter space by again selecting the distribution of $\bm \Theta$. Ultimately, one could change the objective function in each gradient step to be something even more suitable for parameter estimation, like the deterministic \gls{crb}.

\bibliographystyle{IEEEtran}
\bibliography{paper/references}

\end{document}